\newif\ifdoublecol

\doublecoltrue
\documentclass[twocolumns]{IEEEtran}
\IEEEoverridecommandlockouts

\usepackage{xr}

\usepackage{etex} 

\ifCLASSINFOpdf
  \usepackage[pdftex]{graphicx}

  \usepackage{epstopdf}
\else
  \usepackage[dvips]{graphicx}
  \graphicspath{{../eps/}}
  \DeclareGraphicsExtensions{.eps}
\fi
\usepackage{cite}

\usepackage{algorithmicx}

\usepackage[cmex10]{amsmath}
\usepackage{xcolor,colortbl}
\definecolor{col1}{HTML}{3891A6}
\definecolor{col2}{HTML}{EF5B5B}
\definecolor{col3}{HTML}{3DDC97}

\usepackage{amsmath,amsthm,relsize}
\usepackage{amsfonts, amssymb, cuted}
\usepackage{mathtools}
\usepackage{algorithm}
\usepackage[noend]{algpseudocode}
\usepackage{cite}
\usepackage{soul}
\usepackage{graphicx}

\usepackage{tikz}
\usepackage{pgfplotstable}
\usepackage{pgfplots}
\usepackage{nicefrac}
\usepackage{enumitem}
\usepackage{support-caption}
\usepackage{subcaption}
\usepackage[font=footnotesize]{caption}
\usepackage{multirow}
\pgfplotsset{compat=1.15}
\usepackage{relsize}
\usetikzlibrary{patterns}
\usepackage{acro}
\usepackage{pifont}

\usepackage{todonotes}

\usepackage{mathtools}

\usepackage{array}
\usepackage{booktabs}

\usepackage{diagbox}
\usepackage{smartdiagram}
\usesmartdiagramlibrary{additions}
\usepackage{bm}
\usepackage{multicol}

\usepackage{tabularx}
\usepackage{colortbl}

\usepackage{hyperref}
\hypersetup{
    colorlinks=true,
    linkcolor=blue,
    citecolor=blue,
    filecolor=blue,
    urlcolor=blue
}

\usepackage{cleveref}

\usepackage{geometry}

\usetikzlibrary{plotmarks}
  \usetikzlibrary{arrows.meta}
  \usepgfplotslibrary{patchplots}
  \usepackage{grffile}
  \pgfplotsset{plot coordinates/math parser=false}
  \newlength\figureheight
  \newlength\figurewidth
   \pgfplotsset{compat=1.11,
    /pgfplots/ybar legend/.style={
    /pgfplots/legend image code/.code={%
       \draw[##1,/tikz/.cd,yshift=-0.25em]
        (0cm,0cm) rectangle (3em,8pt);},
   },
}

\usepgfplotslibrary{groupplots,units}
\pgfplotsset{
  compat=1.9,
  unit code/.code 2 args={\si{#1#2}} 
}
\usepackage{siunitx}

\usepackage{algpseudocode}

\usepackage{color}
\usepackage{colortbl}
\usepackage{multirow}

\newlength{\Oldarrayrulewidth}

\definecolor{intnull}{RGB}{213,229,255}
\definecolor{inteins}{RGB}{128,179,255}
\definecolor{intzwei}{RGB}{42,127,255}
\definecolor{intdrei}{RGB}{0,85,212}
\definecolor{intvier}{RGB}{0,51,128}
\definecolor{intfunf}{RGB}{0,34,85}

\usetikzlibrary{decorations.pathreplacing}

\newtheorem{lemma}{Lemma}
\newtheorem{corollary}{Corollary}
\newtheorem{remarknum}{Remark} 
\usepackage{arydshln}

\newcommand{\herm}{^{\mbox{\scriptsize H}}}

\newcommand{\vbar}{\raisebox{.17ex}{\rule{.04em}{1.35ex}}}
\newcommand{\vbarind}{\raisebox{.01ex}{\rule{.04em}{1.1ex}}}
\newcommand{\R}{\ifmmode{\rm I}\hspace{-.2em}{\rm R} \else ${\rm I}\hspace{-.2em}{\rm R}$ \fi}
\newcommand{\T}{\ifmmode{\rm I}\hspace{-.2em}{\rm T} \else ${\rm I}\hspace{-.2em}{\rm T}$ \fi}
\newcommand{\N}{\ifmmode{\rm I}\hspace{-.2em}{\rm N} \else \mbox{${\rm I}\hspace{-.2em}{\rm N}$} \fi}
\newcommand{\B}{\ifmmode{\rm I}\hspace{-.2em}{\rm B} \else \mbox{${\rm I}\hspace{-.2em}{\rm B}$} \fi}
\newcommand{\Hil}{\ifmmode{\rm I}\hspace{-.2em}{\rm H} \else \mbox{${\rm I}\hspace{-.2em}{\rm H}$} \fi}
\newcommand{\C}{\ifmmode\hspace{.2em}\vbar\hspace{-.31em}{\rm C} \else \mbox{$\hspace{.2em}\vbar\hspace{-.31em}{\rm C}$} \fi}
\newcommand{\Cind}{\ifmmode\hspace{.2em}\vbarind\hspace{-.25em}{\rm C} \else \mbox{$\hspace{.2em}\vbarind\hspace{-.25em}{\rm C}$} \fi}
\newcommand{\Q}{\ifmmode\hspace{.2em}\vbar\hspace{-.31em}{\rm Q} \else \mbox{$\hspace{.2em}\vbar\hspace{-.31em}{\rm Q}$} \fi}
\newcommand{\Z}{\ifmmode{\rm Z}\hspace{-.28em}{\rm Z} \else ${\rm Z}\hspace{-.28em}{\rm Z}$ \fi}

\newtheorem{thm}{Theorem}
\newtheorem{lem}{Lemma}

\newtheorem{exmp}{Example}
\theoremstyle{definition}

\newcommand{\CB}[0]{{\mathcal{B}}}

\newcommand{\CD}[0]{{\mathcal{D}}}
\newcommand{\CE}[0]{{\mathcal{E}}}
\newcommand{\CF}[0]{{\mathcal{F}}}

\newcommand{\CK}[0]{{\mathcal{K}}}
\newcommand{\CL}[0]{{\mathcal{L}}}
\newcommand{\CM}[0]{{\mathcal{M}}}
\newcommand{\CN}[0]{{\mathcal{N}}}

\newcommand{\CP}[0]{{\mathcal{P}}}
\newcommand{\CQ}[0]{{\mathcal{Q}}}

\newcommand{\CT}[0]{{\mathcal{T}}}

\newcommand{\CV}[0]{{\mathcal{V}}}

\newcommand{\Ba}[0]{{\mathbf{a}}}
\newcommand{\Bb}[0]{{\mathbf{b}}}

\newcommand{\Bu}[0]{{\mathbf{u}}}

\newcommand{\Bw}[0]{{\mathbf{w}}}
\newcommand{\Bx}[0]{{\mathbf{x}}}
\newcommand{\By}[0]{{\mathbf{y}}}
\newcommand{\Bz}[0]{{\mathbf{z}}}

\newcommand{\BA}[0]{{\mathbf{A}}}
\newcommand{\BB}[0]{{\mathbf{B}}}

\newcommand{\BH}[0]{{\mathbf{H}}}
\newcommand{\BI}[0]{{\mathbf{I}}}

\newcommand{\BU}[0]{{\mathbf{U}}}

\newcommand{\BW}[0]{{\mathbf{W}}}

\newcommand{\SfB}[0]{{\mathsf{B}}}

\newcommand{\SfE}[0]{{\mathsf{E}}}

\newcommand{\SfP}[0]{{\mathsf{P}}}

\newcommand{\SfS}[0]{{\mathsf{S}}}

\usepackage{tikz}

\usepackage{tikz}
\usetikzlibrary{arrows,
                calc,
                decorations.pathreplacing,
                calligraphy,
                matrix,
                positioning
                }

\DeclareAcronym{ADMM}{
    short = ADMM,
    long = alternating direction method of multipliers,
    list = Alternating Direction Method of Multipliers,
    tag = abbrev
}

\DeclareAcronym{AoA}{
    short = AoA,
    long = angle-of-arrival,
    list = Angle-of-Arrival,
    tag = abbrev
}

\DeclareAcronym{SISO}{
    short = SISO,
    long = single-input single-output,
    list = single-input single-output,
    tag = abbrev
}

\DeclareAcronym{MRT}{
    short = MRT,
    long = maximum ratio transmitter,
    list = maximum ratio transmitter,
    tag = abbrev
}

\DeclareAcronym{PDA}{
    short = PDA,
    long = placement delivery array,
    list = placement delivery array,
    tag = abbrev
}

\DeclareAcronym{EE}{
    short = EE,
    long = energy efficiency,
    list = energy efficiency,
    tag = abbrev
}

\DeclareAcronym{MDS}{
    short = MDS,
    long = maximum distance separation,
    list = maximum distance separation,
    tag = abbrev
}

\DeclareAcronym{SIC}{
    short = SIC,
    long = successive-interference-cancellation,
    list = successive-interference-cancellation,
    tag = abbrev
}

\DeclareAcronym{MAC}{
    short = MAC,
    long = multiple-access-channel,
    list = multiple-access-channel,
    tag = abbrev
}

\DeclareAcronym{AoD}{
    short = AoD,
    long = angle-of-departure,
    list = Angle-of-Departure,
    tag = abbrev
}

\DeclareAcronym{BB}{
    short = BB,
    long = base band,
    list = Base Band,
    tag = abbrev
}

\DeclareAcronym{BC}{
    short = BC,
    long = broadcast channel,
    list = Broadcast Channel,
    tag = abbrev
}

\DeclareAcronym{BS}{
    short = BS,
    long = base station,
    list = Base Station,
    tag = abbrev
}

\DeclareAcronym{BR}{
    short = BR,
    long = best response,
    list = Best Response, 
    tag = abbrev
}

\DeclareAcronym{CB}{
    short = CB,
    long = coordinated beamforming,
    list = Coordinated Beamforming,
    tag = abbrev
}

\DeclareAcronym{CC}{
    short = CC,
    long = coded caching,
    list = Coded Caching,
    tag = abbrev
}

\DeclareAcronym{CE}{
    short = CE,
    long = channel estimation,
    list = Channel Estimation,
    tag = abbrev
}

\DeclareAcronym{CoMP}{
    short = CoMP,
    long = coordinated multi-point transmission,
    list = Coordinated Multi-Point Transmission,
    tag = abbrev
}

\DeclareAcronym{CRAN}{
    short = C-RAN,
    long = cloud radio access network,
    list = Cloud Radio Access Network,
    tag = abbrev
}

\DeclareAcronym{CSE}{
    short = CSE,
    long = channel specific estimation,
    list = Channel Specific Estimation,
    tag = abbrev
}

\DeclareAcronym{CSI}{
    short = CSI,
    long = channel state information,
    list = Channel State Information,
    tag = abbrev
}

\DeclareAcronym{CSIT}{
    short = CSIT,
    long = channel state information at the transmitter,
    list = Channel State Information at the Transmitter,
    tag = abbrev
}

\DeclareAcronym{CU}{
    short = CU,
    long = central unit,
    list = Central Unit,
    tag = abbrev
}

\DeclareAcronym{D2D}{
    short = D2D,
    long = device-to-device,
    list = Device-to-Device,
    tag = abbrev
}

\DeclareAcronym{DE-ADMM}{
    short = DE-ADMM,
    long = direct estimation with alternating direction method of multipliers,
    list = Direct Estimation with Alternating Direction Method of Multipliers,
    tag = abbrev
}

\DeclareAcronym{DE-BR}{
    short = DE-BR,
    long = direct estimation with best response,
    list = Direct Estimation with Best Response,
    tag = abbrev
}

\DeclareAcronym{DE-SG}{
    short = DE-SG,
    long = direct estimation with stochastic gradient,
    list = Direct Estimation with Stochastic Gradient,
    tag = abbrev
}

\DeclareAcronym{DFT}{
	short = DFT,
	long = discrete fourier transform,
	list = Discrete Fourier Transform,
	tag = abbrev
}

\DeclareAcronym{DoF}{
    short = DoF,
    long = degrees of freedom,
    list = Degrees of Freedom,
    tag = abbrev
}

\DeclareAcronym{DL}{
    short = DL,
    long = downlink,
    list = Downlink,
    tag = abbrev
}

\DeclareAcronym{GD}{
	short = GD, 
	long = gradient descent,
	list = Gradeitn Descent,
	tag = abbrev
}

\DeclareAcronym{IBC}{
    short = IBC,
    long = interfering broadcast channel,
    list = Interfering Broadcast Channel,
    tag = abbrev
}

\DeclareAcronym{i.i.d.}{
    short = i.i.d.,
    long = independent and identically distributed,
    list = Independent and Identically Distributed,
    tag = abbrev
}

\DeclareAcronym{JP}{
    short = JP,
    long = joint processing,
    list = Joint Processing,
    tag = abbrev
}

\DeclareAcronym{KKT}{
    short = KKT,
    long = Karush-Kuhn-Tucker,
    tag = abbrev
}

\DeclareAcronym{LOS}{
	short = LOS,
	long = line-of-sight,
	list = Line-of-Sight,
	tag = abbrev
}

\DeclareAcronym{LS}{
    short = LS,
    long = least squares,
    list = Least Squares,
    tag = abbrev
}

\DeclareAcronym{LTE}{
    short = LTE,
    long = Long Term Evolution,
    tag = abbrev
}

\DeclareAcronym{LTE-A}{
    short = LTE-A,
    long = Long Term Evolution Advanced,
    tag = abbrev
}

\DeclareAcronym{MIMO}{
    short = MIMO,
    long = multiple-input multiple-output,
    list = Multiple-Input Multiple-Output,
    tag = abbrev
}

\DeclareAcronym{MISO}{
    short = MISO,
    long = multiple-input single-output,
    list = Multiple-Input Single-Output,
    tag = abbrev
}

\DeclareAcronym{MSE}{
    short = MSE,
    long = mean-squared error,
    list = Mean-Squared Error,
    tag = abbrev
}

\DeclareAcronym{MMSE}{
    short = MMSE,
    long = minimum mean-squared error,
    list = Minimum Mean-Squared Error,
    tag = abbrev
}

\DeclareAcronym{mmWave}{
	short = mmWave,
	long = millimeter wave,
	list = Millimeter Wave,
	tag = abbrev
}

\DeclareAcronym{MU-MIMO}{
    short = MU-MIMO,
    long = multi-user \ac{MIMO},
    list = Multi-User \ac{MIMO},
    tag = abbrev
}

\DeclareAcronym{OTA}{
    short = OTA,
    long = over-the-air,
    list = Over-the-Air,
    tag = abbrev
}

\DeclareAcronym{PSD}{
    short = PSD,
    long = positive semidefinite,
    list = Positive Semidefinite,
    tag = abbrev
}

\DeclareAcronym{QoS}{
	short = QoS,
	long = quality of service,
	list = Quality of Service,
	tag = abbrev
}

\DeclareAcronym{RCP}{
	short = RCP,
	long = remote central processor,
	list = Remote Central Processor,
	tag = abbrev
}

\DeclareAcronym{RRH}{
    short = RRH,
    long = remote radio head,
    list = Remote Radio Head,
    tag = abbrev
}

\DeclareAcronym{RSSI}{
    short = RSSI,
    long = received signal strength indicator,
    list = Received Signal Strength Indicator,
    tag = abbrev
}

\DeclareAcronym{RX}{
	short = RX,
	long = receiver,
	list = Receiver,
	tag = abbrev
}

\DeclareAcronym{SCA}{
    short = SCA,
    long = successive convex approximation,
    list = Successive Convex Approximation,
    tag = abbrev
}

\DeclareAcronym{SG}{
    short = SG,
    long = stochastic gradient,
    list = Stochastic Gradient,
    tag = abbrev
}

\DeclareAcronym{SNR}{
    short = SNR,
    long = signal-to-noise ratio,
    list = Signal-to-Noise Ratio,
    tag = abbrev
}

\DeclareAcronym{SINR}{
    short = SINR,
    long = signal-to-interference-plus-noise ratio,
    list = Signal-to-Interference-plus-Noise Ratio,
    tag = abbrev
}

\DeclareAcronym{SOCP}{
	short = SOCP, 
	long = second order cone program,
	list = Second Order Cone Program,
	tag = abbrev
}

\DeclareAcronym{SSE}{
    short = SSE,
    long = stream specific estimation,
    list = Stream Specific Estimation,
    tag = abbrev
}

\DeclareAcronym{SVD}{
	short = SVD,
	long = singular value decomposition,
	list = Singular Value Decomposition,
	tag = abbrev
}

\DeclareAcronym{TDD}{
	short = TDD,
	long = time division duplex,
	list = Time Division Duplex,
	tag = abbrev
}

\DeclareAcronym{TX}{
	short = TX,
	long = transmitter,
	list = Transmitter,
	tag = abbrev
}

\DeclareAcronym{UE}{
    short = UE,
    long = user equipment,
    list = User Equipment,
    tag = abbrev
}

\DeclareAcronym{UL}{
    short = UL,
    long = uplink,
    list = Uplink,
    tag = abbrev
}

\DeclareAcronym{ULA}{
	short = ULA,
	long = uniform linear array,
	list = Uniform Linear Array,
	tag = abbrev
}

\DeclareAcronym{UPA}{
    short = UPA,
    long = uniform planar array,
    list = Uniform Planar Array,
    tag = abbrev
}

\DeclareAcronym{WMMSE}{
    short = WMMSE,
    long = weighted minimum mean-squared error,
    list = Weighted Minimum Mean-Squared Error,
    tag = abbrev
}

\DeclareAcronym{WMSEMin}{
    short = WMSEMin,
    long = weighted sum \ac{MSE} minimization,
    list = Weighted sum \ac{MSE} Minimization,
    tag = abbrev
}

\DeclareAcronym{WBAN}{
	short = WBAN,
	long = wireless body area network,
	list = Wireless Body Area Network,
	tag = abbrev
}

\DeclareAcronym{WSRMax}{
    short = WSRMax,
    long = weighted sum rate maximization,
    list = Weighted Sum Rate Maximization,
    tag = abbrev
}

\begin{document}

\title{
Cache-Aided MIMO Communications:\\DoF Analysis and Transmitter Optimization 
}


\author{\IEEEauthorblockN{Mohammad NaseriTehrani,
MohammadJavad Salehi,
~and Antti T\"olli}
\thanks{
The authors are affiliated with the University of Oulu, Finland. Emails: \{firstname.lastname@oulu.fi\}.
This work was supported by Infotech Oulu and by the Research Council of Finland under grants no. 343586 (CAMAIDE) and 346208 (6G Flagship).
This article has been presented in part in~\cite{naseritehrani2024enhanced} and~\cite{naseritehrani2024multicast}.\color{black}
}. 


}

\maketitle


\begin{abstract}
Cache-aided MIMO communications aims to jointly exploit both coded caching~(CC) and spatial multiplexing gains to enhance communication efficiency.  
In this paper, we analyze both the achievable degrees of freedom~(DoF) under linear processing constraint and the finite-SNR performance of a MIMO-CC system with CC gain \(t\), where a server with \(L\) transmit antennas communicates with \(K\) users, each equipped with \(G\) receive antennas. We first demonstrate that the enhanced DoF of 
\(\max_{\beta, \Omega} \Omega \times \beta\) is achievable with linear processing, where the number of users \(\Omega\) served in each transmission is fine-tuned to maximize DoF, and \(\beta \le \min\big(G, \nicefrac{L \binom{\Omega-1}{t}}{\big(1 + (\Omega - t - 1)\binom{\Omega-1}{t}}\big)\big)\) represents the number of parallel streams decoded by each user. 
Then, we propose a new class of MIMO-CC schemes using a novel scheduling mechanism leveraging maximal multicasting opportunities to maximize delivery rates at given SNR levels while still adhering to linear processing constraints. 
This new class of schemes is paired with an efficient linear multicast beamformer design, resulting in a more practical, high-performance solution for integrating CC in future MIMO systems.
\end{abstract}

\begin{IEEEkeywords}
\noindent Coded caching, MIMO communications, scheduling, beamforming, degrees of freedom
\end{IEEEkeywords}

\section{Introduction}
\label{section:intro}
Mobile data traffic is continuously growing due to exponentially increasing volumes of multimedia content and the rising popularity of emerging applications such as mobile immersive viewing and extended reality~
\cite{rajatheva2020whiteb,salehi2022enhancing}.
The existing wireless network infrastructure is under considerable strain due to the particularly demanding requirements of these applications, ranging from high throughput to ultra-low latency data delivery. 
 This has motivated the development of new innovative techniques, among which,
coded caching~(CC), originally proposed in the pioneering work~\cite{maddah2014fundamental}, is particularly promising as it offers a new degree-of-freedom~(DoF) gain that scales proportionally to the cumulative cache size across all network users. In fact, CC enables the use of the onboard memory of network devices as a new communication resource, appealing especially for multimedia applications where the content is cacheable by nature~\cite{salehi2022enhancing,akcay2023optimal,bayat2021coded}. 
To enable this new gain,
in the so-called placement phase, content from a library of files is proactively stored in the receiver caches. This is then followed by a delivery phase, where carefully built codewords are multicast to groups of users of size $t+1$, where the CC gain $t\equiv\frac{KM}{N}$ represents the cumulative cache size across all $K$
users, each with a cache memory large enough to store $M$ files, normalized by the library size of $N$ files.
The codewords are built such that each user can eliminate undesired parts of the message using its cache contents.
%
%
%
Later, the original CC scheme of~\cite{maddah2014fundamental} was extended to more diverse network conditions and topologies, including multi-server~\cite{shariatpanahi2016multi}, wireless~\cite{tolli2017multi,shariatpanahi2018physical,salehi2019subpacketization}, D2D~\cite{ji2015fundamental,mahmoodi2023d2d}, shared-cache~\cite{parrinello2019fundamental,parrinello2020extending}, multi-access~\cite{serbetci2019multi}, dynamic~\cite{abolpour2024cache}, content-aware~\cite{Mahmoodi2023Multi-antennaDelivery}, and combinatorial~\cite{brunero2022fundamental} networks.

To explore the application of CC in wireless networks comprehensively, it is imperative to investigate the specific attributes of the wireless medium, encompassing its broadcast nature, channel fading, and varying interference. 
This is especially true in the context of multi-antenna systems, given their prominent importance in enabling higher throughput in communication systems~\cite{rajatheva2020whiteb}. In this regard,
the theoretical and practical dimensions of applying CC in multi-input single-output~(MISO) setups have undergone comprehensive exploration in prior research~\cite{tolli2017multi,tolli2018multicast,shariatpanahi2018physical,salehi2019subpacketization, lampiris2021resolving}. In contrast, only a few works have addressed the integration of multiple-input multiple-output~(MIMO) techniques with CC solutions, primarily focusing on enhancing the total DoF measured in terms of the number of simultaneously delivered parallel streams in the network~\cite{salehi2021MIMO,salehi2023multicast,cao2017fundamental,cao2019treating}. Still, many theoretical and practical aspects of applying CC techniques in MIMO systems remain largely unexplored.  

\subsection{Related Work}
\label{section:intro_rela_w}

Existing works on cache-aided multi-antenna communications with CC techniques have pursued three major goals: increasing the achievable DoF, enhancing finite-SNR performance, and resolving the subpacketization bottleneck.

\subsubsection{Achievable DoF analysis} 
Early works on the integration of the original CC scheme~\cite{maddah2014fundamental} in multi-antenna communications targeted downlink {MISO} setups, revealing the interesting fact that with the CC gain of $t$ and $L$ transmit antennas, $t+L$ users can be served in parallel. In other words, in MISO-CC systems, the total DoF of $t+L$ is achievable~\cite{shariatpanahi2018physical}, and is optimal under simple constraints~\cite{lampiris2021resolving}. These studies were later extended to {MIMO} setups. In~\cite{cao2017fundamental}, the optimal DoF of cache-aided MIMO networks with three transmitters and three receivers was studied, and in~\cite{cao2019treating}, general message sets were used to introduce two inner and outer bounds on the achievable DoF of MIMO-CC schemes. However, achieving the DoF bounds required complex interference alignment techniques. 
More recently, a new MIMO-CC scheme was introduced in~\cite{salehi2021MIMO}, where the MIMO system was interpreted as an extension of the shared-cache setup developed for MISO systems~\cite{parrinello2019fundamental,parrinello2020extending}, and it was shown that with the CC gain $t$, $L$ antennas at the transmitter, and $G$ antennas at each receiver, when $L$ is divisible by $G$, the single-shot DoF of $Gt+L$ is achievable with a small subpacketization overhead. While the extension mechanism in~\cite{salehi2021MIMO} provides a straightforward solution to build MIMO-CC schemes using shared-cache models originally designed for MISO systems, the resulting DoF remains below that of more advanced scheduling-based solutions. 
Moreover, the resulting schemes will necessarily rely on cache-aided interference cancellation in the signal domain~\cite{salehi2022enhancing}.
In other works on MIMO-CC systems, a high-DoF transmission framework for cache-aided MIMO interference networks was designed in~\cite{liu2024coding}, and a partially connected shared-cache network with distributed single-antenna helpers jointly serving single-antenna users was studied in~\cite{akcay2025collaborative}. In the latter work, the overall network was modeled as a MIMO Gaussian broadcast channel, enabling a two-phase delivery scheme leveraging both CC and spatial multiplexing gains.

\subsubsection{Finite-SNR analysis}
Pioneering works on the DoF analysis of both MISO- and MIMO-CC systems~\cite{shariatpanahi2018physical,salehi2021MIMO} relied on zero-force (ZF) beamforming at the transmitter (and matched filtering at the receivers, in the context of MIMO systems). To address the inefficient finite-SNR performance of the ZF beamforming in the MISO-CC scheme in~\cite{shariatpanahi2018physical},
an optimized design of multi-group multicast beamformers was proposed in~\cite{tolli2017multi,tolli2018multicast}. In the same works, the spatial multiplexing gain
and the number of partially overlapping multicast messages were flexibly adjusted to find an appropriate trade-off between reduced design complexity and improved finite-SNR performance. As an alternative approach, a simple iterative solution exploiting Lagrangian duality to design optimized beamformers was proposed in~\cite{mahmoodi2021low}. In the context of MIMO system, in~\cite{salehi2023multicast}, the authors developed optimized unicast and multicast beamformers tailored for the scheme in~\cite{salehi2021MIMO}. In particular, the multicast beamformer design was based on decomposing the system into multiple parallel MISO setups (for divisible $L/G$) where several multicast codewords could be transmitted simultaneously. More recently, a high-performance but highly complex covariance-based multi-group multicasting design for MIMO-CC systems was introduced in~\cite{naseritehrani2024multicast}. 

\subsubsection{Subpacketization bottleneck}
Subpacketization reflects the division of each file into smaller parts for the CC operation~\cite{lampiris2018adding}. Both the original single-antenna and MISO-CC schemes of~\cite{maddah2014fundamental,shariatpanahi2018physical} required exponentially growing subpacketization (w.r.t the user count $K$), rendering them infeasible for even moderate-sized networks~\cite{lampiris2018adding}. To resolve this issue, the pioneering work in~\cite{lampiris2018adding} introduced signal-level CC operation, where (part of) the interference is regenerated from the local memory and is eliminated from the received signal before decoding at the receiver~\cite{salehi2022enhancing} (in contrast, the original MISO-CC scheme~\cite{shariatpanahi2018physical} relied on bit-level processing by multicasting carefully created XOR codewords to multiple user groups while suppressing the remaining inter-stream interference through spatial processing~\cite{shariatpanahi2018physical,tolli2017multi}). The work in~\cite{lampiris2018adding} showed that, through signal-level processing, the same optimal DoF of $t + L$ could be achieved in MISO-CC setups with much smaller subpacketization. 
Later, the cyclic scheme proposed in~\cite{salehi2020lowcomplexity} also employed signal-level interference cancellation to achieve linearly growing subpacketization, further improving scalability.
However, the reduced subpacketization in both schemes comes at the cost of limited applicability, as~\cite{lampiris2018adding} imposes divisibility constraints on the system parameters, requiring that both 
$\tfrac{L}{t}$ and $\tfrac{K}{t}$ are integers, and~\cite{salehi2020lowcomplexity} is applicable only to MISO-CC setups with $L \ge t$.

More recently, signal-level CC has also proven effective in addressing several practical bottlenecks of conventional CC. Most prominently, signal-level schemes allow simpler optimized beamformer designs by enabling dedicated unicast beamformers for each data stream~\cite{salehi2021low}, and facilitate extending CC applicability to use cases with location-dependent file requests~\cite{Mahmoodi2023Multi-antennaDelivery, mahmoodi2024low} and dynamic user mobility~\cite{abolpour2024cache,abolpour2024resource}. However, there is a noticeable performance loss in terms of the achievable finite-SNR rate due to the lack of multicast beamforming gain available in the bit-level approach~\cite{salehi2019subpacketization,salehi2022multi}.
%
In signal-level interference cancellation, each receiver reconstructs the interfering terms from its cached content, requiring PHY-layer cache access and related control signaling. This is similar to superposition coding with successive interference cancellation~\cite{ding2017application,dai2018survey}, but somewhat simpler since the interfering symbols are locally known, avoiding error propagation and decoding order constraints~\cite{salehi2022enhancing}.
\color{black}
\subsection{Main Contributions}
\label{section:intro_contrib}
In this work, we propose a novel CC-based content delivery framework that integrates MIMO systems with CC, well suited for a broad range of multimedia applications with cacheable content, ranging from collaborative multi-user XR to video streaming services that require high data rate connectivity with low latency. 
In particular, we study the integration of CC and MIMO connectivity under both asymptotic (high-SNR) and finite-SNR regimes, through unified theoretical analysis and practical algorithmic design. The resulting insights and designs represent a significant advancement beyond existing works in the literature.
This paper includes several contributions, falling under two major categories:


\subsubsection{DoF analysis} We study the asymptotic performance of MIMO-CC systems in Section~\ref{section:DoF}, by analyzing the fundamental achievable DoF under linear decodability constraints.
The analysis is done through introduction of a signal-level MIMO-CC scheme in Theorem~\ref{Th:linear_achievablity}, where instead of serving a fixed number of users in each transmission, we judiciously select the number of users and the spatial multiplexing order per user in order to maximize the DoF. This design provides greater flexibility in selecting system parameters and eliminates the integer constraint on $\frac{L}{G}$ imposed by~\cite{salehi2021MIMO}, resulting in an enhanced DoF of $\max_{\Omega,\beta}{\Omega \times \beta}$ larger than or equal to the DoF of $Gt+L$ in~\cite{salehi2021MIMO}, where $\Omega$ represents the number of users served in each transmission, and $\beta$, where
\begin{equation*}
    \beta \le \mathrm{min}\left(G,\tfrac{L \binom{\Omega-1}{t}}{1 + (\Omega - t-1)\binom{\Omega-1}{t}}\right),
\end{equation*}
denotes the total number of parallel data streams received by each user.

The DoF analysis in this section is an extension of our earlier conference publications in~\cite{naseritehrani2024enhanced,naseritehrani2024multicast}, with a more detailed description of the delivery algorithm and  a clearer correctness verification to ensure that the numbers of missing and delivered subpackets match
(Theorem~\ref{Th:linear_achievablity}).
Specifically, the improved achievable DoF value in MIMO-CC systems under linear decodability constraints was first proposed in~\cite{naseritehrani2024multicast} as a conjecture, without any proof (the main contribution of that work was not DoF analysis but to introduce a high performance non-linear covariance-based transmission design). The first formal achievability proof of the conjecture in~\cite{naseritehrani2024multicast} was later presented in~\cite{naseritehrani2024enhanced}, by introducing a new CC scheme with cache-aided interference cancellation in the signal domain.

In this section, we have further complemented the DoF analysis by multiple new contributions w.r.t to~\cite{naseritehrani2024multicast, naseritehrani2024enhanced}, including a stand-alone, scheme-agnostic linear decodability condition (Theorem~\ref{Th:DoF}), a one-dimensional search algorithm to find the optimized DoF (Corollary~\ref{remark1-Dof}), identification of structural limitations in the solution space for naive candidate selection (Lemma~\ref{lm:beta_G_NBachievable}), a tie-breaking rule connecting asymptotic and finite-SNR regimes by prioritizing among DoF optimal pairs based on their symmetric-rate performance (Remark~\ref{underloading}), and DoF gap analysis with state-of-the-art (Lemma~\ref{lm:wsa_Prop_achievable_DOF_gaps}).

\subsubsection{Finite-SNR analysis} Building on the DoF insights and recognizing the implementation difficulty and reduced multicasting gain of signal-level schemes, we investigate the finite-SNR performance of MIMO-CC systems in Section~\ref{section:DoF_MC} by introducing bit-level interference cancellation and full-size XOR transmission to the linear decodability constraint, and proposing a completely new class of scheduling algorithms built upon advancements in hypergraph theory (Theorems~\ref{Th:Base_scheduling} and~\ref{Th:Scheduling Mechanism}). The goal is to enable flexible design of the delivery algorithm, given the SNR value, to improve the symmetric rate. 
 In this part, we also introduce a non-trivial extension of the MISO-CC scheme in~\cite{shariatpanahi2016multi} to MIMO-CC setups as another baseline to be compared with the new class of schemes, referred to as the Ext-MS scheme (Section~\ref{section:DoF_MC}). Furthermore, as a minor contribution, we develop an iterative linear beamforming solution that integrates into the scheduling scheme and builds upon the solution in~\cite{mahmoodi2021low} but is tailored for the new class of symmetric schemes by accommodating partially overlapping codewords while ensuring linear decodability at the users (relegated to Appendix~\ref{section:Linear_BF}). The  result is a simple yet efficient solution for enabling the gain boost of CC in practical MIMO systems, where maximizing the DoF is not necessarily the primary design objective.
 %

Extensive numerical simulations show the improved performance of our proposed solution, from both DoF and symmetric rate perspectives, with respect to state-of-the-art. In particular, in the finite-SNR regime, the proposed framework achieves superior performance compared to state-of-the-art linear schemes and approaches the performance of the non-linear design in~\cite{naseritehrani2024multicast}.

\color{black}

\textbf{Notations.}
Throughout the text, $(\cdot)^H$ and $(\cdot)^{-1}$ denote the Hermitian and inverse of a matrix, respectively. Let $\mathbb{C}$ and $\mathbb{N}$ denote the sets of complex and natural numbers.
For integer $J$, $[J] \equiv \{1,2,\dots,J\}$, 
for vectors $\Ba$, $\Bb$, $\cdots$, $[\Ba\;\Bb\;\dots]$ denotes their horizontal concatenation, and for matrices $\BA$, $\BB$, $\cdots$, $[\BA\;\BB\;\dots]$ represents their horizontal concatenation. \color{black} 
 Boldface upper- and lower-case letters indicate matrices and vectors, respectively, and calligraphic letters denote sets. $|\CK|$ denotes the cardinality of the set $\CK$, and $\CK \backslash \CT$ is the set of elements in $\CK$ that are not in $\CT$. Supersets are denoted by sans-serif letters, and $|\SfB|$ indicates the size of a superset $\SfB$.
Additionally, $\oplus$ denotes the XOR operation over a finite field.

\section{System Model}
\label{section:sys_model}
A MIMO setup is considered, in which a single BS equipped with \(L\) transmit antennas serves \(K\) cache-enabled users, each having \(G\) receive antennas, as shown in Figure~\ref{fig:ISIT_sysmo}.\footnote{
In fact, $L$ and $G$ represent the spatial multiplexing gain at the transmitter and receivers, respectively, which may be less than the actual number of antennas depending on the channel rank and the number of baseband RF chains. Nevertheless, the term `antenna count' is used for simplicity throughout the text.
}
Every user has a cache memory of size $MF$ bits, and requests a single, unique file from a library $\CF$ of $N$ files, each with the size of $F$ bits. 
Without loss of generality, we assume a normalized data unit and omit the file size \(F\) in the subsequent notations. 
The coded caching gain is defined as $t \equiv \frac{KM}{N}$, representing how many replicas of the file library can be stored collectively across the cache memories of all users. In this paper, we assume $K \ge t+1$. The system operation comprises two phases: placement and delivery. In the placement phase, the users' cache memories are filled with data. Following a similar structure as~\cite{shariatpanahi2016multi}, we split each file $W \in \CF$ into $\binom{K}{t}$ packets $W_{\CP}$, where $\CP \subseteq [K]$ denotes any subset of users with $|\CP| = t$. 
Then, we store each packet $W_{\CP}$ in the cache of user $k \in [K]$ 
if and only if $k \in \CP$. In other words, user $k$ stores all packets of all files 
whose index sets include $k$, i.e., 
$\{\, W_{\CP} : W \in \CF,\; \CP \subseteq [K],\; k \in \CP \,\}$. \color{black}

\begin{figure}[t]
        \centering
        \includegraphics[height = 5.8cm]{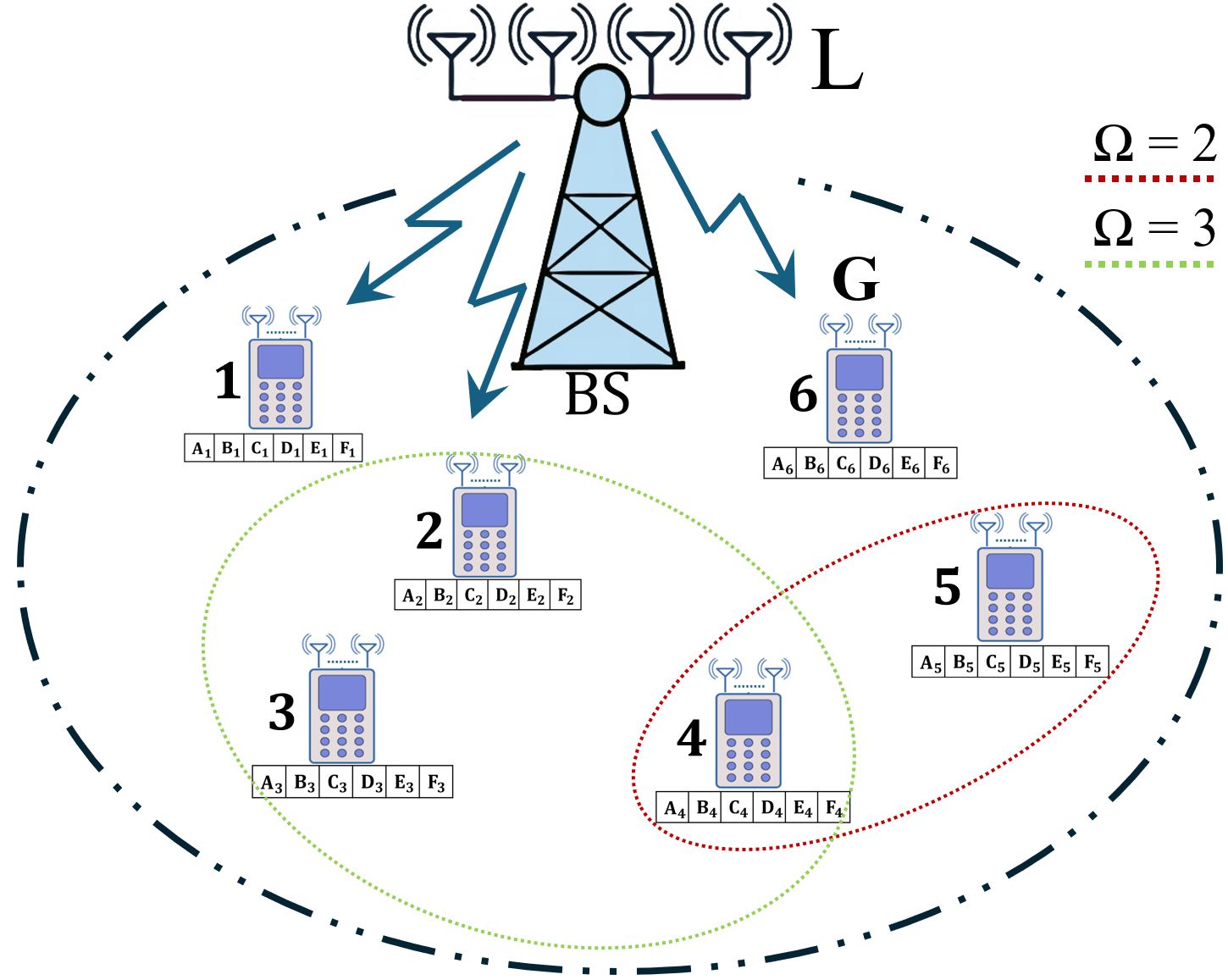} \label{fig:receiver_Prop}
    \caption{MIMO-CC system model and user selection for different $\Omega$.
    }
    \label{fig:ISIT_sysmo}
\end{figure}

At the beginning of the delivery phase, each user $k$ reveals its requested file $W_k \in \CF$ to the server. 
Then, for every subset of users $\CK \subseteq [K]$ with size $|\CK| = \Omega$, the server creates $S$ transmission vectors $\Bx_{\CK}(s)$, $s \in [S]$, each delivering parts of the requested data to every user $k \in \CK$. Here, $t+1 \le \Omega \le K$, is a design parameter and $S$ is a constant multiplier defined by the specific delivery algorithm. In other words, data delivery is done through a total number of $S \binom{K}{\Omega}$ vectors, which are transmitted, e.g., in consecutive time intervals.
\color{black}
\color{black} Let us now consider a transmission vector $\Bx_{\CK}(s)$ delivering data to a subset $\CK$ of users with $|\CK| = \Omega$.
Upon transmission of $\Bx_{\CK}(s)$, user $k \in \CK$ receives

\begin{equation}\label{eq:RX_signal}
\By_k(s) = \BH_k(s) \Bx_{\CK}(s) + \Bz_k(s) \color{black}\; ,
\end{equation}
where $\BH_k(s)\color{black} \in \mathbb{C}^{G \times L}$ represents the channel matrix between the server and user $k$, and $\Bz_k(s)\color{black} \sim \mathcal{CN}(\mathbf{0},N_0 \mathbf{I})$ is the noise, in interval $s$. The entries of $\BH_k(s)$ are considered independent identically distributed (i.i.d) Gaussian variables with zero-mean and unit variance, and full channel state information (CSI) is assumed to be available at the server.\footnote{In practical downlink scenarios, we commonly use Time-Division Duplex (TDD) for uplink-downlink transmissions. In this setup, the BS estimates downlink channels by leveraging uplink pilot transmissions through channel reciprocity~\cite{tolli2019distributed}.}
We adopt a standard block-fading model, in which channel realizations remain constant within a coherence interval and change independently across intervals according to user mobility. Accordingly, the transmitter re-optimizes its beamformers at the beginning of each coherence interval based on the newly acquired CSI, while the higher-level caching and scheduling structure remains unchanged. \color{black}
\begin{remarknum}
The exact meaning of a `file' in our model depends on the use case. For example, in video-on-demand (VoD) applications, each large multimedia file is divided into smaller `chunks'~\cite{akcay2023optimal,bayat2021coded}, and each chunk is treated as an independent file in the CC formulation. Users sequentially request the chunks of their requested video stream as playback progresses.
The timeline can thus be divided into consecutive `delivery frames,' each roughly matching the chunk duration (on the order of seconds).
Within each delivery frame, the latest chunks requested by all active users are jointly delivered using multicast transmission, so assuming that $K$ users request cacheable files ``at the same time'' simply means they are active within the same frame. Also, since chunks are relatively small and the CC delivery covers multiple transmissions within each frame, the transmission intervals naturally align with the block-fading timescale.
\end{remarknum}
\color{black}
In order to define the symmetric rate, we need to know the length (in data units) of each transmission vector. Based on the delivery algorithm, each packet $W_{\CP}$ may need to be further divided into a number of equal-sized subpackets before constructing the transmission vectors. Let us use $\Theta$ to represent the final subpacketization level, encompassing the splitting factor in both the placement and delivery phases. As will be demonstrated, each transmission vector corresponds to a new set of subpackets sent in parallel. 
Using $R_{\CK}(s)\color{black}$ (file/second) to denote the max-min transmission rate of $\Bx_{\CK}(s)\color{black}$ ensuring successful decoding
at every user $k \in \CK$ in interval $s$, the transmission time of $\Bx_{\CK}(s)\color{black}$ is $T_{\CK}(s) = \tfrac{1}{(\Theta R_{\CK}(s))}\color{black}$ (seconds). Let us denote the total delivery time (the sum of $T_{\CK}(s)\color{black}$ for all user subsets $\CK$ and interval indices $s$) with $T_{\mathrm{total}}$. Then, the symmetric rate is defined as $R_{sym} = \tfrac{K}{T_{\mathrm{total}}}$~(file/second), and
the goal is to design the delivery scheme to maximize $R_{sym}$.

Throughout this paper, 
depending on the considered delivery scheme, the transmission vector $\Bx_{\CK}(s)\color{black}$ may comprise unicast or multicast signals.
For the `multicast' transmission, the individual data terms are first added (XOR'd) in the bit domain, and then, the modulated XOR signals are served to the users with multicast beamformers. However,
for the `unicast' transmission, the user-specific modulated signals are first multiplied by corresponding unicast beamformers and then superimposed in the complex (signal) domain to form the transmission vector~\cite{salehi2022enhancing}. 
The exact composition of the transmission vector $\Bx_{\CK}(s)$ will be detailed later as we introduce each delivery algorithm.

\section{Achievable DoF Analysis}
\label{section:DoF}




In this section, we take a closer look at the {linear decodability} and achievable DoF\footnote{Here, the term DoF is used equivalent to the total number of parallel spatial dimensions delivered in each transmission~\cite{salehi2020lowcomplexity,lampiris2018adding,salehi2021MIMO}.} of MIMO-CC setups. Compared to the state-of-the-art analysis in~\cite{salehi2021MIMO}, our scheme exceeds its achievable DoF of $Gt+L$, and eliminates its restrictive integer divisibility constraint $\tfrac{L}{G}\in\mathbb{N}$. 
\color{black} 
%
%
Following the system model in Section~\ref{section:sys_model}, for every subset $\CK$ of users with size $\Omega$, we build $S$ transmission vectors $\Bx_{\CK}(s)\color{black}$, each delivering parts of the requested data to all users in $\CK$. Assume $\beta$ parallel streams are delivered to each user $k \in \CK$ in each transmission $\Bx_{\CK}(s)\color{black}$. The goal is to maximize the total number of streams per transmission (i.e., $\Omega \times \beta$) while assuring linear decodability by each target user. In this section, we consider a signal-level design, and we also assume
zero-forcing~(ZF) beamformers are employed at both the transmitter and receiver sides to null out the inter-user and inter-stream interference, respectively.\footnote{Both ZF beamformers and signal-level interference cancellation are assumed to demonstrate the achievability of the proposed DoF at high SNR. 
In Section~\ref{section:DoF_MC}, we introduce a new class of MIMO-CC schemes with bit-level interference cancellation. 
Furthermore, for more practical communication at finite SNR, optimized multicast beamformer design is introduced in Appendix~\ref{section:Linear_BF}. 
\color{black}
}
In this regard, each packet \( W_{\mathcal{P},k} \) of the file \( W_k \) requested by a user \( k \) is further split into a number of smaller subpackets \( W_{\mathcal{P},k}^q \) (\( q \) is the subpacket index -- the number of subpackets is clarified shortly), and the signal-level transmission vector \(\Bx_{\CK}(s) \color{black}\) 
is modeled as 

\begin{equation}\label{sig_model_uc1}   
    \Bx_{\CK}(s) = \sum_{i \in \CK}  \sum_{W_{\CP,i}^q \in \CM_{i}(s)}   \Bw_{\CP,i}^q W_{\CP,i}^q,\color{black}
\end{equation}  
where \(\CM_i(s)\color{black} \) denotes the set of subpackets intended for user \( i \) in interval $s$ (all file fragments are considered as modulated signals for simplicity), and \( \Bw_{\CP,i}^q \) represents the corresponding transmit beamforming vector.
%
Before proceeding to the main results, let us review an intuitive example.


\begin{exmp}
\label{exmp:x2serv}
\normalfont For a setup with $K=3$, $L = 3$, $G = 2$, $t = 1$, and $\Omega = 3$, we show that in every transmission, $\beta = G = 2$ parallel data streams can be linearly decoded by each target user. 
In the placement phase, each file is split into $\binom{K}{t} = 3$ packets, and each user stores one packet of each file.
For example, if library files are shown by $A, B, C,\cdots$, user~1 stores packets $A_{\{1\}}, B_{\{1\}}, C_{\{1\}}, \cdots$, where the size of each packet is $1/3$ of the original file.

At the beginning of the delivery phase, assume users~1-3 request files $A$-$C$, respectively (i.e., $W_{1}=A$, $W_{2}=B$, and $W_{3}=C$). In this particular example, we have only one subset $\CK$ with size $\Omega = 3$, $S = 1$, and we do not also need an additional level of subpacketization. So we can ignore $\CK$, $s$, and $q$ indices in subsequent notations\color{black}. 
The transmission vector $\Bx$ is designed as

\begin{equation*}
\begin{aligned}
\Bx =  \Bw_{\{2\},1} A_{\{2\}} + \Bw_{\{1\},2} B_{\{1\}} +  \Bw_{\{3\},1} A_{\{3\}} \\[.5ex] \qquad+ \Bw_{\{1\},3} C_{\{1\}}   +  \Bw_{\{3\},2} B_{\{3\}} + \Bw_{\{2\},3} C_{\{2\}},  
\end{aligned}
\end{equation*}
where, for example, the beamformer vectors $\Bw_{\{2\},1}$ and $\Bw_{\{1\},2}$ are projected to the null space of user~3 such that no inter-stream interference is caused to user~3 by $A_{\{2\}}$ and $A_{\{1\}}$, respectively.

Let us now consider the decoding process 
by user~1, which receives
    $\By_1 = \BH_1 \Bx + \Bz_1$.
Assuming equivalent channels $\BH_1 \Bw_{\{1\},2}$ and $\BH_1 \Bw_{\{1\},3}$ 
are estimated from the downlink precoded pilots, 
the interference terms $\BH_1 \Bw_{\{1\},2} B_{\{1\}}$ and $ \BH_1 \Bw_{{\{1\}},3} C_{\{1\}}$ can be first reconstructed and removed from the received signal as $\tilde{\By}_1 = \By_1 - \BH_1 \Bw_{{\{1\}},2} B_{\{1\}} - \BH_1 \Bw_{{\{1\}},3} C_{\{1\}}$. The remaining received signal vector $\tilde{\By}_1$ is then multiplied by the receive beamforming vectors  $\BU_1 = $ $\big[\Bu_{1,1}, \Bu_{1,2}\big] \in \mathbb{C}^{2 \times 2}$:

\begin{align*}
y_{1,1} &= \Bu_{1,1}\herm \BH_1 \Bw_{\{2\},1} A_{\{2\}} 
+ \Bu_{1,1}\herm \BH_1 \Bw_{\{3\},1} A_{\{3\}} \\
& + \Bu_{1,1}\herm \BH_1 \Bw_{\{3\},2} B_{\{3\}} 
+ \Bu_{1,1}\herm \BH_1 \Bw_{\{2\},3} C_{\{2\}} + z_{1,1} \\[2ex]
y_{1,2} &= \Bu_{1,2}\herm \BH_1 \Bw_{\{2\},1} A_{\{2\}} 
+ \Bu_{1,2}\herm \BH_1 \Bw_{\{3\},1} A_{\{3\}} \\
& + \Bu_{1,2}\herm \BH_1 \Bw_{\{3\},2} B_{\{3\}} 
+ \Bu_{1,2}\herm \BH_1 \Bw_{\{2\},3} C_{\{2\}} + z_{1,2}
\end{align*}

where $z_{1,i} = \Bu_{1,i}\herm \Bz_1, \ i =1,2$. For user~1 to decode $A_{\{2\}}$ from $y_{1,1}$ and $A_{\{3\}}$ from $y_{1,2}$, we enforce

\begin{align*}
\Bu_{1,i}\herm \BH_1 \Bw_{\{3\},2} &= 0, \quad 
\Bu_{1,i}\herm \BH_1 \Bw_{\{2\},3} = 0, \quad i \in [2] \\[0.9mm]
 \Rightarrow \Bw_{\{3\},2}, &\Bw_{\{2\},3} 
\in \mathrm{Null}\Big(\big[\BH_1\herm\Bu_{1,1}, \BH_1\herm\Bu_{1,2}\big]\herm\Big) \\[1ex]
\Bu_{1,1}\herm \BH_1 \Bw_{\{3\},1} &= 0 
\Rightarrow \Bu_{1,1} \in \mathrm{Null}\big(\BH_1 \Bw_{\{3\},1}\big) \\[1ex]
\Bu_{1,2}\herm \BH_1 \Bw_{\{2\},1} &= 0 
\Rightarrow \Bu_{1,2} \in \mathrm{Null}\big(\BH_1 \Bw_{\{2\},1}\big)
\end{align*}


where $\mathrm{Null}(\cdot)$ denotes the null space. These conditions are satisfied as the dimensions of 
$[\BH_1\herm\Bu_{1,1}, \BH_1\herm\Bu_{1,2}]\herm$ and
$\BH_1 \Bw_{\CP,1}$, $\CP \in \{{\{2\}},{\{3\}}\}$ are $2 \times 3$ and $2 \times 1$, respectively (in fact, for this particular setup, \(\Bw_{{\{3\}},2}=\Bw_{{\{2\}},3}\), \(\Bw_{{\{2\}},1}=\Bw_{{\{1\}},2}\), and \(\Bw_{{\{3\}},1}=\Bw_{{\{1\}},3}\). However, this is not true for the general case).
Similar conditions hold for successful decoding by users~$2$ and~$3$, and so, each user can linearly decode $\beta = 2$ parallel streams, and the total DoF of 6 is achievable.
 \hfill $\qed$


\end{exmp}

For the general case, given the transmission model in~\eqref{sig_model_uc1} and the received signal model in~\eqref{eq:RX_signal}, after the transmission of $\Bx_{\CK}(s)\color{black}$, a user $k \in \CK$ receives

\begin{equation}
   \By_{k}(s) = \sum\limits_{i \in \CK}  \sum\limits_{W_{\CP,i}^q \in \CM_{i}(s)}  \BH_{k}(s)  \Bw_{\CP,i}^q W_{\CP,i}^q + \Bz_k(s)\color{black} .
\end{equation}
In order to decode its intended $\beta$ subpackets in $\CM_{k}(s)\color{black}$, user $k$ applies a receive beamforming matrix $\BU_{k}(s)\color{black} \in \mathbb{C}^{G \times \beta}$ to $\By_{k}(s)\color{black}$. By definition, a subpacket $W_{\CP,k}^q \in \CM_{k}(s)\color{black}$ with packet index denoted by $\CP$ transmitted to user $k$ is already available in the cache memory of every user $i \in \CP$, and these users can reconstruct and remove the interference caused by transmissions of $W_{\CP,k}^q$. So, for interference-free decoding of $\Bx_{\CK}(s)\color{black}$, the ZF beamforming vector $\Bw_{\CP,k}^q$ should null out the interference caused by $W_{\CP,k}^q$ from every stream decoded by each user $i \in \CK \backslash (\CP \cup \{k\})$. 
The following theorem establishes a necessary condition for the linear decodability of $\Bx_{\CK}(s)\color{black}$ in~\eqref{sig_model_uc1}.

\color{black}
\begin{thm}\label{Th:DoF}
For the considered MIMO-CC setup, to ensure linear decodability at each user $k \in \CK$,
the number of streams per user (i.e., $\beta$) must satisfy:

\begin{equation} \label{DoF_bndd}
    \beta \le \min \left(G, \big(L-(\Omega-t-1)\beta\big)\binom{\Omega-1}{t}\right).
\end{equation}
\end{thm}
\begin{proof}

Let us define the equivalent interference channel of a subpacket $W_{\CP,k}^q$ included in $\Bx_{\CK}(s)$ as
\begin{equation}
\label{eq:equi_intf_chan}
\bar{\BH}_{\CP,k}(s)   = \Big[\,\BH_{i}\herm(s)\BU_{i}(s)\,\Big]\herm, \quad \forall i \in \CK \backslash (\CP\cup \{k\})\color{black},
\end{equation}
where $[\;\cdot\;]$ represents the horizontal concatenation of matrices inside the brackets. Note that the definition in~\eqref{eq:equi_intf_chan} is agnostic to the subpacket index $q$ and depends only on the user index $k$ and the packet index $\CP$.
The ZF beamforming vector $\Bw_{\CP,k}^q$ must null out the inter-stream interference caused by $W_{\CP,k}^q$ to every stream decoded by each user $i \in \CK \backslash (\CP \cup \{k\})$ (the users in $\CP$ can remove the interference with their cache contents). This condition implies that

\begin{equation}\label{eq:beamforming_in_nullspace}
    \Bw_{\CP,k}^q \in \mathrm{Null}\big(\bar{\BH}_{\CP,k}(s)\big).\color{black} 
\end{equation}
To guarantee that user~$k$ can decode all its $\beta$ parallel streams $W_{\CP,k}^q \in \CM_{k}(s)\color{black}$ in a linear fashion, it is necessary that $\beta \le G$ as the dimensions of the channel matrix $\BH_{k}(s)\color{black}$ are $G \times L$. However, linear decodability also necessitates that the transmit beamformers $\Bw_{\CP,k}^q$ are linearly independent. For subpackets with different packet indices \( \CP \), this condition could be met as the beamformers are chosen from non-coinciding null spaces corresponding to different (either non-overlapping or partially overlapping) user sets (c.f~\eqref{eq:beamforming_in_nullspace}). 
However, for successfully decoding of subpackets with the same packet index $\CP$, the number of such subpackets should be constrained by the dimensions of $\mathrm{Null}\big(\bar{\BH}_{\CP,k}(s)\big)\color{black}$.
As the total number of parallel streams per user is $\beta$, $|\CP| = t$, and $\CP \subseteq \CK \backslash \{k\}$, we have at least $\Big\lceil \tfrac{\beta}{\binom{\Omega-1}{t}}\Big\rceil$ subpackets in $\CM_{k}(s)\color{black}$ with a similar packet index $\CP$. 
On the other hand, from~\eqref{eq:beamforming_in_nullspace}, the beamforming vectors $\Bw_{\CP,k}^q$ are chosen from the null space of $\bar{\BH}_{\CP,k}(s)\color{black}$, which, from~\eqref{eq:equi_intf_chan}, is constructed by concatenating $\Omega-t-1$ matrices $\BH_{k}\herm(s)\BU_{k}(s)\color{black}$, each of size $L \times \beta$. As a result, $\bar{\BH}_{\CP,k}(s)\color{black} \in \mathbb{C}^{(\Omega-t-1)\beta \times L}$. Using 
$\mathrm{nullity(\cdot)}$ to denote the dimensions of the null space, we can use the rank-nullity theorem~\cite{meyer2023matrix} to write

\begin{equation}
 \label{eq:rank-null-theorem}
 \begin{aligned}
\mathrm{nullity}\big(\bar{\BH}_{\CP,k}(s) \big)  &= L - \mathrm{rank}\big(\bar{\BH}_{\CP,k}(s) \big) \\[1ex]&= L - ({\Omega-t-1})\beta.
\end{aligned}
 \end{equation}
Thus, for successful decoding of subpackets with the same packet index, it is necessary that

\begin{equation}
\label{eq:ceil_ineq_dof}
    \left\lceil  \frac{ \beta}{\binom{\Omega-1}{t}} \right\rceil \le L - (\Omega-t-1) \beta.
\end{equation}
Since the right-hand-side of~\eqref{eq:ceil_ineq_dof} is an integer,~\eqref{eq:ceil_ineq_dof} can be equivalently written as
\begin{equation}
    \beta \le \left(L-(\Omega-t-1)\beta \right) \binom{\Omega-1}{t},
\end{equation}
which, together with the basic decoding criteria of $\beta \le G$, results in~\eqref{DoF_bndd}.
\end{proof}
\begin{thm}\label{Th:linear_achievablity}
    For every pair \((\Omega,\beta)\) satisfying~\eqref{DoF_bndd} in Theorem~\ref{Th:DoF}, there exists a linearly decodable signal-level coded caching scheme with the DoF of $\Omega \times \beta$.
\end{thm}
\begin{proof}
In the following, a generalized linear scheme with the DoF of $\Omega \times\beta$ is provided.
The \textit{placement phase} is performed as detailed in Section~\ref{section:sys_model}: each file $W \in \CF$ is split into $\binom{K}{t}$ packets $W_{\CP}$ and each user $k \in [K]$ stores a packet $W_{\CP}$ if $k \in \CP$.

In the \textit{delivery phase}, each packet \(W_{\mathcal{P},k}\) of the file $W_k$ requested by a user $k$ is further split into \(\beta \binom{K-t-1}{\Omega-t-1}\) smaller subpackets \(W_{\mathcal{P},k}^q\), with $q\in \Big[\beta \binom{K-t-1}{\Omega-t-1} \Big]$.\footnote{Here, our goal is only to prove the ``achievability'' of the proposed DoF value. In this regard, the subpacketization level is chosen such that it satisfies the requirements in general while it is not necessarily at the minimal level for some specific parameter combinations.}
Then, for every subset \(\mathcal{K}\) of users with the cardinality of \(|\mathcal{K}| = \Omega\),  \(S \triangleq \binom{\Omega-1}{t}\) transmission vectors \(\mathbf{x}_{\CK}(s)\color{black}\), \(s \in [S]\) are constructed with Algorithm~\ref{alg:alg_1}. 
In lines~\ref{algline_firstphase_start}-\ref{algline_firstphase_end} of the algorithm, for each user \(k \in \mathcal{K}\), we choose $\beta \times S$ fresh subpackets of its requested file $W_k$. Then, in lines~\ref{algline_secondphase_start}-\ref{algline_secondphase_end}, for each transmission interval \( s \in [S] \), we construct a transmission vector \(\Bx_{\CK}(s)\color{black}\) that delivers $\beta$ subpackets to each user in $\CK$. To show the correctness of the algorithm, we first investigate the linear decodability of $\Bx_{\CK}(s)\color{black}$ and then show that all the missing subpackets are delivered.

\begin{algorithm}[t]
\small
\caption{Constructing transmission vectors \(\Bx_{\CK}(s)\)\color{black}} \label{alg:alg_1}

\begin{algorithmic}[1]
\ForAll{$k \in \CK$} \label{algline_firstphase_start}
    \State $\SfP_k \gets \{\CP \subseteq  \CK \backslash \{k\}, |\CP| = t\}$
    \ForAll{$\CP \in \SfP_k$}
        \State $\CL_{\CP,k} \gets \{W_{\CP,k}^q \mid W_{\CP,k}^q \textrm{ is not delivered} \}$
        \State $\CN_{\CP,k} \gets$ a subset of $\CL_{\CP,k}$ with size $\beta$
    \EndFor
\EndFor \label{algline_firstphase_end}
\ForAll{$s \in [S]$} \label{algline_secondphase_start}
    \ForAll{$k \in \CK$}
        \State $\CM_k(s) \gets \varnothing$
        \While{$|\CM_k(s)| < \beta$}
            \State $\CP^* \gets \arg\max_{\CP} |\CN_{\CP,k}|$ \label{algline_max_cardinality}
            \State $\hat{W}^{q}_{\CP^*,k} \gets$ a subpacket from $\CN_{\CP^*,k}$
            \State $\CM_k(s) \gets \CM_k(s) \cup \{\hat{W}^{q}_{\CP^*,k}\}$
            \State $\CN_{\CP^*,k} \gets \CN_{\CP^*,k} \backslash \{\hat{W}^{q}_{\CP^*,k}\}$
        \EndWhile
    \EndFor
    \State $\Bx_{\CK}(s) \gets \sum_{k \in \CK}  \sum_{W_{\CP,k}^q \in \CM_{k}(s)}   \Bw_{\CP,k}^q W_{\CP,k}^q\color{black}$ \label{algline_transvec}
\EndFor \label{algline_secondphase_end}
\end{algorithmic}
\end{algorithm}

\noindent\textbf{Linear decodability:}
For each user $k \in \CK$ and each interval $s \in [S]$, selecting the $\CN_{\CP,k}$ set with the largest cardinality in line~\ref{algline_max_cardinality} of the algorithm ensures that the number of subpackets in $\CM_k(s)\color{black}$ with a similar packet index is minimized. This is because, after selecting a particular $\CN_{\CP',k}$ set and moving one of its subpackets to $\CM_k(s)\color{black}$, $|\CN_{\CP',k}|$ is decremented by one, and so this set will not be selected again until $|\CN_{\CP,k}| \le |\CN_{\CP',k}|$ for all $\CP \in \SfP_k \backslash \{\CP'\}$. As a result, as $|\SfP_k| = S$ and also because $|\CM_k(s) |\color{black} = \beta$ before line~\ref{algline_transvec} of the algorithm, the number of subpackets of $W_k$ with a similar packet index transmitted by $\Bx_{\CK}(s)\color{black}$ is upper-bounded by \(\big\lceil \frac{\beta}{S} \big\rceil\), and linear decodability is guaranteed by Theorem~\ref{Th:DoF}.

\noindent\textbf{Missing packet delivery:}
Each user needs $\binom{K-1}{t}$ packets of its requested file $W_k$ (the rest are cached in its memory), and during the delivery phase, each of these requested packets is split into \( \beta\binom{K-t-1}{\Omega-t-1} \) subpackets. Thus, the total number of missing subpackets per user is given by 
\begin{equation*}
   \binom{K-1}{t}\binom{K-t-1}{\Omega-t-1}\beta.
\end{equation*}
On the other hand, the subpackets of a packet \( W_{\CP,k} \) are included in a transmission vector \( \Bx_{\CK}(s)\color{black} \) only if user \( k \) is in the set of target users \( \CK \) and \( \CP \subset \CK \). Clearly, since \( |\CP| = t \) and \( |\CK| = \Omega \), the number of sets \( \CK \) satisfying these constraints is 
\begin{equation*}
    \binom{K-1}{\Omega-1}\binom{\Omega-1}{t}.
\end{equation*}
Furthermore, for every such set \( \CK \), \( \beta \) subpackets of \( W_{\CP,k} \) are delivered using the respective transmission vectors \( \Bx_{\CK}(s)\color{black} \). Consequently, the total number of missing subpackets per intended user, and across all users for each requested packet, exactly matches the number of delivered subpackets of that packet.\end{proof}
\begin{remarknum}\label{subpkt_order}
To maintain generality and applicability across arbitrary network parameters $K$, $t$, $L$, and $G$, the schemes developed in this paper follow a combinatorial structure similar to that of the MISO-CC scheme in~\cite{tolli2017multi}. While this has enabled us to fully characterize the achievable single-shot DoF gains and finite-SNR performance of MIMO-CC systems, the resulting schemes inevitably require a large subpacketization level, hindering their applicability to networks with a large number of users.
Nevertheless, the same extension principles used here to generalize the combinatorial structure in~\cite{tolli2017multi} to MIMO setups can be readily applied to other low-subpacketization MISO-CC designs as well~\cite{lampiris2018adding,salehi2020lowcomplexity,wan2022multiple}. This approach has already been taken in~\cite{salehi2021MIMO}, where the cyclic scheme in~\cite{salehi2020lowcomplexity} was extended to MIMO setups with linearly growing subpacketization. 
However, such low-subpacketization designs inevitably introduce applicability constraints of the underlying MISO-CC solution; for instance, the design in~\cite{salehi2021MIMO} is valid only when $\big\lfloor \tfrac{L}{G} \big\rfloor \ge t$, a condition inherited from the cyclic scheme~\cite{salehi2020lowcomplexity}.
\end{remarknum}
\color{black}
\begin{corollary} \label{remark1-Dof}
\normalfont The DoF of $\beta \times \Omega$ is necessarily achievable in every given MIMO setup, as long as $\beta$ and $\Omega$ satisfy~\eqref{DoF_bndd}. As a result, the maximum achievable DoF for the proposed MIMO-CC transmission design is given by solving
\color{black}
\begin{equation}\label{eq:total_DoF}
\begin{aligned} 
 & {\mathrm{DoF}_{}}(\beta^* , \Omega^*) =  \max_{\beta , \Omega }~ \Omega \times \beta, \\[1ex]
   &\mathrm{s.t.}\:\:{\beta \le \mathrm{min}\bigg({G}, \frac{L \binom{\Omega-1}{t}}{ 1 + (\Omega - t-1)\binom{\Omega-1}{t}} \bigg),}\\[1ex]
    & t+1 \leq \Omega \leq t+L, \quad \Omega \in \mathbb{Z}_+, \quad \beta \in \mathbb{Z}_+.
\end{aligned}
\end{equation}

To find the optimized parameters $\beta^*$ and $\Omega^*$, we first impose an explicit constraint that the largest feasible 
$\beta$ is chosen for each $\Omega$ while maximizing the DoF in~\eqref{eq:total_DoF}:

\begin{equation}\label{eq:beta_omega_relation}
\beta= \Bigg\lfloor \min\bigg(G, \frac{L \binom{\Omega-1}{t}}{1 + (\Omega - t-1)\binom{\Omega-1}{t}}\bigg)\Bigg\rfloor,
\end{equation}
and then simply determine the maximum achievable DoF by searching over \(\Omega = \!t+\!1\) to \(\!t+\!L\) as

\begin{equation}
\begin{array}{l}
  \Omega^* =  \arg\max\limits_{\substack{\\[1ex] t+1 \le \Omega \le t+L\\[1ex] \Omega \in \mathbb{Z}_+ }} \Omega  \Bigg\lfloor \min\Big(G, \frac{ L \binom{\Omega-1}{t}}{1 + (\Omega - t-1)\binom{\Omega-1}{t}}\Big)\Bigg\rfloor.\\[2mm]
\end{array}
\end{equation}  
Plugging the resulting \(\Omega^*\) into~\eqref{eq:beta_omega_relation} yields \(\beta^*\), and the optimized DoF is given as $\mathrm{DoF}_{} = \Omega^* \times \beta^*$.

 \end{corollary}

\begin{lem}\label{lm:beta_G_NBachievable}

When fully utilizing the receiver-side multiplexing gain (i.e., $\beta = G$), linear decoding requires
\begin{equation}\label{eq:Omega_bnd_beta_eq_G}
       \Omega \leq \Big\lfloor \dfrac{L}{G} \Big\rfloor + t + 1 .
\end{equation}
In other words, choosing $\beta = G$ limits the range of possible values for $\Omega$, and hence, the maximum achievable DoF may be less than the jointly optimized DoF in~\eqref{eq:total_DoF}.

\end{lem} 

\noindent \textit{Proof.}
From Theorem~\ref{Th:DoF},
linear decodability requires:
    \begin{equation} \label{DoF_bnd_betaG}
    \displaystyle (\Omega - t - 1) + \frac{1}{\binom{\Omega - 1}{t}} \leq \frac{L}{\beta}.
    \end{equation}
To show that the feasible solution space for \(\Omega\) is limited by~\eqref{eq:Omega_bnd_beta_eq_G}, assume, for contradiction, that a larger value  
\begin{equation}
\Omega = \Big\lfloor \dfrac{L}{G} \Big\rfloor + t + 2
\end{equation}
is also valid for \eqref{DoF_bnd_betaG}. Substituting this and $\beta=G$ into \eqref{DoF_bnd_betaG} gives:
\begin{equation}
\Big\lfloor \dfrac{L}{G} \Big\rfloor + 1 + \dfrac{\displaystyle 1}{\displaystyle\binom{\Big\lfloor \dfrac{L}{G} \Big\rfloor + t + 1}{t}} \leq \frac{L}{G},
\end{equation}
which, since \( \binom{t+n}{t} \ge 1, \forall n \in \mathbb{N} \), necessitates that:
\begin{equation} \label{ineq_floor}
\Big\lfloor \dfrac{L}{G} \Big\rfloor + 1 < \frac{L}{G}.
\end{equation}
However, this cannot be valid due to the properties of the floor function, and hence, the largest possible value for $\Omega$ is $\Big\lfloor \tfrac{L}{G}\Big \rfloor + t + 1$.

\begin{figure}[t]
    \centering
    \resizebox{\columnwidth}{!}{%

   
    \begin{tikzpicture}
    \begin{axis}[
        axis lines = left,
        xlabel = \smaller {Number of users served per Tx. interval $\Omega$},
        ylabel = \smaller {Achievable DoF},
        ylabel near ticks,
        xlabel near ticks,
        legend pos = north west,
        legend columns=5,
        legend style = {
            nodes={scale=1.0, transform shape},
            cells={align=center},
            at={(0.58,0.95)},
            anchor=north west,
            font=\tiny,
        },
        ybar interval=0.99,
        ymin = 0,
        ymax = 26,
        xmin = 2,
        ticklabel style={font=\smaller},
        grid=both,
    major grid style={line width=.2pt,draw=gray!30},
    ]

    \addplot
    [black,fill = gray!70]
    table
    [x=Omega,y=dofG4]
    {Figs/DoF_prop_behav.tex};
    \addlegendentry{\small $G=4$}

    \addplot 
    [black,fill = blue!30]
    table
    [x=Omega,y=dofG8]
    {Figs/DoF_prop_behav.tex};
    \addlegendentry{\small $G=8$}
    
    

    




    
    \end{axis}
    \end{tikzpicture}
    }
    \caption{Behavior of the solution to \eqref{eq:total_DoF} in Corollary~\ref{remark1-Dof} for \(L \)=16, \(t\)=1. }
    \label{fig:plot_behav}
\end{figure}

\begin{exmp}\label{exmp:dof_behavior} 
Assume $L = 16$ and $t = 1$. In Figure~\ref{fig:plot_behav}, we plot the maximum achievable DoF under linear decodability for different values of $\Omega$, considering two scenarios: $G \in \{4, 8\}$.
For $G = 8$, the maximum DoF of $24$ is achieved with $\Omega^* = 3$ and $\beta^* = G = 8$. This corresponds to one of the two optimal solutions.
In contrast, for $G = 4$, achieving the maximum DoF of $21$ requires selecting $\Omega^* = 7$ and $\beta^* = 3$. Imposing $\beta = G = 4$ in this case would limit the feasible $\Omega$ to at most $6$, which in turn constrains the achievable DoF to $20$ (Lemma~\ref{lm:beta_G_NBachievable}).
\end{exmp}

\begin{remarknum}\label{underloading}
There may exist multiple pairs of \((\Omega,\beta)\) that result in the same DoF while satisfying the linear decodability constraints of Theorem~\ref{Th:DoF}. In fact, as can be seen for the $G=8$ case in Example~\ref{exmp:dof_behavior}, even the optimized $\textrm{DoF}_{}$ from solving~\eqref{eq:total_DoF} may be achievable by multiple choices of $(\Omega^*,\beta^*)$.
In such cases, the transmitter side load---defined as the total number of parallel streams, i.e., $\Omega \times \beta$---is identical across all candidate solutions. Among these, in finite SNR, we prioritize the solutions that spread the streams across users as much as possible, thereby maximizing the spatial degrees of freedom at the receivers.
The resulting increase in the null-space dimensions available at each receiver expands the feasible solution space and enables more flexibility for the system to jointly optimize Tx and Rx beamformers---i.e., to maximize the desired terms while suppressing inter-user and inter-stream interference. As a result, the SINR per user improves, directly enhancing the symmetric rate. In the case of the network in Example~\ref{exmp:dof_behavior}, this means selecting $(\Omega^*,\beta^*) = (4,6)$ over $(\Omega^*,\beta^*) = (3,8)$, as in the former case, the receivers are not fully loaded (i.e., $\beta^* < G$), and hence, there is more freedom in designing receive beamformer to enhance the symmetric rate.
\end{remarknum}

\begin{lem}\label{lm:wsa_Prop_achievable_DOF_gaps}
The achievable DoF of $G(t+\big\lfloor \tfrac{L}{G} \big\rfloor)$ in~\cite{salehi2021MIMO} is always less than or equal to the achievable DoF in Corollary~\ref{remark1-Dof}. Nevertheless, the DoF gap between the two schemes is 
at most \( 2(G-1) \).  
\end{lem}

\begin{proof}
    Let $L \triangleq nG + r$, {with } $n \geq 1$ { and } $0 \leq r < G$, 
    and $\Omega \triangleq t + b + 1$, {with } $0 \leq b < L$.
    From~\eqref{DoF_bnd_betaG} we have:
    \begin{align}
\mathrm{DoF} &= (t + b + 1)\beta \label{eq:total_DoF12} \\[.5ex]
\text{s.t.} \quad 
\text{i)} \: &\beta \le G, \: \: \text{ii)} \:\beta \le \big(L - b\beta\big)\binom{t + b}{t} \nonumber
\end{align}
Now, we can examine the maximum DoF gap between \(\mathrm{DoF}_{}\) and \(\mathrm{DoF}_{\text{\cite{salehi2021MIMO}}} = G\big(t+\big\lfloor \tfrac{L}{G}\big\rfloor\big)\) as follows:

    \begin{equation*}\label{eq:total_DoF2}
 \begin{array}{l}
  {\mathrm{DoF}_{}}-{\mathrm{DoF}_{\!\!\text{\cite{salehi2021MIMO}}}}  = (t+b+1) \beta-Gt-\Big\lfloor \dfrac{L}{G}\Big\rfloor G, \\[.1ex]
  \hspace{23mm}\overset{(a)}{\leq}  
  \beta (t+1)-Gt+L - \Big\lfloor \dfrac{L}{G}\Big\rfloor G-1 \\[.1ex] 
   \hspace{23mm}\overset{(b)}{\leq} \beta (t+1)-Gt+G-2 \\[.1ex]
   \hspace{23mm}\overset{(c)}{\leq} G (t+1)-Gt+G- 2 = 2G-2.
 \end{array}
\end{equation*}
Here, \((a)\) follows from the inequality \(b\beta \leq L-1\), obtained from the necessary condition for linear decodability in Theorem~\ref{Th:DoF}, 
\((b)\) follows from \(L - \big\lfloor \tfrac{L}{G}\big\rfloor G = r \leq G-1\), and \((c)\) follows from the necessary condition \(\beta \leq G\) for linear decodability in Theorem~\ref{Th:DoF}.
\end{proof}\color{black}
\begin{exmp}
\label{exmp:x0serv}
\normalfont 
The delivery algorithm in the proof of Theorem~\ref{Th:linear_achievablity} is reviewed in a particular network setup 
with $(L,G,t)=(6,4,1)$. Assume $(\Omega,\beta) = (3,4)$, ensuring linear decodability according to Theorem~\ref{Th:DoF}.
With this selection, we need $S=\binom{\Omega-1}{t}=2$ transmissions per each selection of the target user set $\CK$ with $|\CK| = 3$. Let us consider the first transmission (i.e., $s=1$) for the user set $\CK = \{1,2,3\}$, and
use $A$, $B$, $C$ to denote the files requested by users~1-3, respectively.
According to Algorithm~\ref{alg:alg_1}, the first step is to select $\beta S = 8$ subpackets for each user $k \in \CK$. Clearly, there is only one choice for the supersets of packet indices:
$\SfP_1 = \{{\{2\}}, {\{3\}}\}$, $\SfP_2 = \{{\{1\}}, {\{3\}}\}$, and $\SfP_3 = \{{\{1\}}, {\{2\}}\}$.
However, depending on the number of remaining undelivered subpackets in $\CL_{\CP,k}$ per each packet index $\CP$, we may have multiple choices for the sets of subpackets $\CN_{\CP,k}$ (as $\CN_{\CP,k} \subseteq \CL_{\CP,k}$ and $|\CN_{\CP,k}| = \beta= 4$). Without loss of generality, let us assume $\CN_{\CP,k} = \{W_{\CP,k}^1,\allowbreak W_{\CP,k}^2,\allowbreak W_{\CP,k}^3,\allowbreak W_{\CP,k}^4\}$ for all $k \in \CK$ and $\CP \in \SfP_k$.

The next step is to select $\beta = 4$ subpackets for each user~$k$. Without loss of generality, assume that by following Algorithm~\ref{alg:alg_1}, we select $\CM_1(1)\color{black}= \{A_{\{2\}}^1$, $A_{\{2\}}^2$, $A_{\{3\}}^1$, $A_{\{3\}}^2\}$, $\CM_2(1)\color{black}= \{B_{\{1\}}^1$, $B_{\{1\}}^2$, $B_{\{3\}}^1$, $B_{\{3\}}^2\}$, and $\CM_3(1)\color{black}= \{ C_{\{1\}}^1$,  $C_{\{1\}}^2$, $C_{\{2\}}^1$, $C_{\{2\}}^2\}$, resulting in the transmission vector: 
\begin{align*}
\Bx_{\CK}(1)\color{black}&= \Bw_{\{2\},1}^1 A_{\{2\}}^1 + \Bw_{\{2\},1}^2 A_{\{2\}}^2 
+ \Bw_{\{3\},1}^1 A_{\{3\}}^1 \\
&\quad + \Bw_{\{3\},1}^2 A_{\{3\}}^2 + \cdots 
+ \Bw_{\{2\},3}^2 C_{\{2\}}^2
\end{align*}

where, for example, the transmit beamforming vector $\Bw_{{\{2\}},1}^{1}$ is designed
to null out the interference caused by $A_{\{2\}}^1$ to the reception of data streams requested by user~3 (i.e., the subpackets in $\CM_3(1)\color{black}$), which, 
Using~\eqref{eq:equi_intf_chan} and~\eqref{eq:beamforming_in_nullspace}, translates to
    $\Bw_{{\{2\}},1}^{1} \in \mathrm{Null}\big([\BH_3(1) \BU_3(1)]\herm$\big).\color{black}

Now, let us review the linear decoding process at user~1, which receives
    $\By_1(1) = \BH_1 \Bx_{\CK}(1) + \Bz_1(1)\color{black}$. 
%
By definition, the interference from $B_{\{3\}}^{1}$, $B_{\{3\}}^{2}$, $C_{\{2\}}^1$, and $C_{\{2\}}^2$ is removed over every stream sent to user~1 using beamforming vectors $\Bw_{{\{3\}},2}^q$ and $\Bw_{{\{2\}},3}^q$, $q\in\{1,2\}$. On the other hand, user~1 has $B_{\{1\}}^q$ and $C_{\{1\}}^q$, $q\in\{1,2\}$, cached in its memory, so it can reconstruct and remove their respective interference terms from
$\By_1(1)$. 
Finally, for a fixed $q \in \{1,2\}$, $\Bw_{{\{2\}},1}^q$ and $\Bw_{{\{3\}},1}^q$ can be designed to be linearly independent as they are chosen from different null spaces ($\mathrm{Null}\big([\BH_3\herm(1) \BU_3(1)]\herm\big)\color{black}$ and $\mathrm{Null}\big([\BH_2\herm(1) \BU_2(1)]\herm\big)$\color{black}, respectively), and for a fixed $\CP \in \{{\{2\}},{\{3\}}\}$, $\Bw_{\CP,1}^1$ and $\Bw_{\CP,1}^2$ can also be selected to be orthogonal as the rank of each null space is given by $\mathrm{nullity\big(\bar{\BH}_{\CP,1}(1)\big)}\color{black} = 6-4 =2$.
%
So, decoding all of the intended data terms $A_{\{2\}}^1$, $A_{\{2\}}^2$, $A_{\{3\}}^1$ and $A_{\{3\}}^2$ is possible at user~1 using the receiver-side ZF beamforming matrix $\BU_1(1) \color{black}\in \mathbb{C}^{4 \times 4}$, designed to suppress any relevant inter-stream interference. 
Similarly, users~2 and~3 can each linearly decode four streams, 
and the total DoF of 12 is achievable.
\end{exmp}

\color{black}



\section{Linear Multicast Transmission Schemes for MIMO-CC}
\label{section:DoF_MC}

As discussed in Section~\ref{section:DoF}, the maximum MIMO-CC DoF under the linear decodability constraints of Theorem~\ref{Th:DoF} can be achieved using unicast beamforming combined with signal-domain interference cancellation.
However, signal domain processing imposes implementation challenges~\cite{salehi2022enhancing}, and 
relying fully on unicast beamforming severely degrades the finite-SNR performance\color{black}.
Similarly, maximizing the number of parallel streams to match the DoF may not even be desirable for rate optimization in finite-SNR~\cite{tolli2017multi,salehi2022multi}.
%
In this section, we introduce a new class of generalized linear multicast transmission strategies that may not necessarily achieve the maximum number of parallel streams (similar to the enhanced DoF value in~\eqref{eq:total_DoF}) but are designed to maximize the delivery rate at a given SNR level constrained by linear processing conditions at each receiver. All the proposed strategies are based on the original multi-server (MS) scheme in~\cite{shariatpanahi2018physical}, take advantage of maximal multicasting opportunities (i.e., XOR codewords of size $t+1$), and are symmetric in the sense that each target user receives an equal number of streams per each transmission.
The linear beamforming used to realize the proposed scheduling framework in MIMO-CC class follows an iterative design adapted from~\cite{mahmoodi2021low}; to keep the focus on the novel scheduling scheme, its details are relegated to Appendix~\ref{section:Linear_BF}. \color{black}
\begin{remarknum}
    The proposed class of linear multicast transmission schemes is a subset of all possible schemes for a given network. The symmetric rate achieved through these schemes may not be globally \color{black} optimal, and, for example, non-linear or non-symmetric schemes with better performance may be found. However, as it is practically impossible to consider all feasible transmission strategies, we focus only on a subset of strategies with a well-defined structure and realistic practical implementability.
\end{remarknum}

We start by reviewing the original MS scheme in~\cite{shariatpanahi2018physical} and assuming that the number of users in each transmission is set to $\Omega$. With this scheme, 
in the delivery phase, each requested packet $W_{\CP,k}$ is further split into $\binom{K-t-1}{\Omega-t-1}$ subpackets denoted as $W_{\CP,k}^q$, and for each subset $\CK$ of users with $|\CK| = \Omega$, a particular transmission vector $\Bx_{\CK,\mathrm{MS}}$ is constructed as:
\begin{equation} \label{sigconst_dof_loss}
    \displaystyle \Bx_{\CK,\mathrm{MS}} = \sum\limits_{\CT \in \SfS^{\CK}} \Bw_{\CT} X_{\CT},
\end{equation}
where $\CT \subseteq \CK$ represents a codeword index, and
\begin{equation}
\label{eq:index_superset}
    {\SfS}^{\CK} = \big\{\CT \subseteq \CK, |\CT| = t+1 \big\}
\end{equation}
denotes the superset of requested codeword indices. Moreover,  
     $X_{\CT} = \bigoplus_{k \in \CT} W_{\CT \backslash \{k\},k}^q$ 
represents a codeword (recall that $W_{\CT \backslash \{k\},k}^q$ denotes a subpacket of the file $W_k$ requested by user $k$)\color{black}, and $\Bw_{\CT}$ is the multicast beamformer vector associated with $X_{\CT}$. The super index $q$ increases sequentially and is used to avoid the repetition of subpackets.

The first option for a cache-enabled MIMO system is to apply the MS scheme directly, i.e., to build the transmission vectors 
similarly as~\eqref{sigconst_dof_loss} but to use spatial multiplexing at each receiver to separate the parallel streams. 
{Throughout the rest of the paper, we call this solution the Extended Multi-Server (Ext-MS) scheme. It can be easily verified that the number of parallel streams per user in the Ext-MS scheme is $\binom{\Omega-1}{t}$, and following Theorem~\ref{Th:DoF}, its linear decodability requires that} 
\begin{equation}
    \begin{array}{l}
    \displaystyle \binom{\Omega-1}{t} \le G, \quad \binom{\Omega-1}{t} \cdot (\Omega - t - 1) \leq L - 1 .
    \end{array}
    \label{eq:ldextms}
\end{equation}
%
%
While the Ext-MS scheme is an easy and straightforward extension of the MS scheme, it faces two critical challenges.
First, if $G < \binom{\Omega-1}{t}$, linear receiver processing is not possible, and the complex successive interference cancellation (SIC) structure is needed to decode the parallel streams. Second, if $G \gg \binom{\Omega-1}{t}$, the solution is very inefficient as the number of streams decoded by each user is much smaller than the maximum possible value (i.e., $G$). To address both scenarios, our proposed schemes introduce an underlying scheduling mechanism that enables setting the number of parallel streams sent to each target user, indicated by $\beta$, to any number from a predefined set 
while maintaining the linear decodability. 
In other words, for each $\Omega$, we first find the set $\CB_{\Omega}$ such that for any $\beta \in \CB_{\Omega}$,
we can build a symmetric linear transmission strategy that transmits $\beta$ parallel streams to each of the $\Omega$ users in $\CK$ in each transmission 
using codewords of size $t+1$ while also ensuring linear decodability.
Then, for a given $\mathrm{SNR}$ value, we pick $\Omega$ and $\beta \in \CB_{\Omega}$ values that maximize the symmetric rate as: 
\begin{equation} \label{symm_rate_general}
     \displaystyle\max\limits_{\Omega, \beta \in \CB_{\Omega}} \displaystyle R_{\mathrm{sym}}\big(\Omega,\beta,\mathrm{SNR}\big).
\end{equation}
\subsection{Enhanced Multicast Scheduling for MIMO-CC}\label{sec:scheduling}

Let us define 

\begin{equation} 
\begin{array}{l}
   \displaystyle \beta_0 = \frac{t+1}{\mathrm{gcd}(t+1,\Omega)}, \quad B_0 = \frac{\Omega}{\mathrm{gcd}(t+1,\Omega)},\\[1ex]
\end{array}
\end{equation}
where $\mathrm{gcd}(\cdot)$ denotes the greatest common divisor.
We first introduce a base scheduling where each target user receives exactly $\beta_0$ codewords in each transmission. This is done in Theorem~\ref{Th:Base_scheduling} using an appropriate partitioning of the codeword index superset ${\SfS}^{\CK}$ as defined in~\eqref{eq:index_superset}\color{black}. Then, in Theorem~\ref{Th:Scheduling Mechanism}, we extend the base scheduling to suggest a more general set of possible $\beta$ values.  

\begin{lemma}
For any given $t$ and $\Omega$, $|\SfS^{\CK}| = \binom{\Omega}{t+1}$ and $\binom{\Omega-1}{t}$ are divisible by $B_0$ and $\beta_0$, respectively.
\end{lemma}
\begin{proof}
The proof follows the
Bézout's identity (or Bézout's lemma) in number theory~\cite{bezout1779theorie}, which asserts that the gcd 
of two integers can be expressed as a linear combination of them 
with integer coefficients.
Using this lemma, we can write 
\begin{equation}
    \gcd(\Omega, t+1) = a \Omega + b (t+1)
\end{equation}
for two integers $a$ and $b$, and as a result
 \begin{align} \label{beta_strm1}
\frac{\displaystyle\binom{\Omega}{t+1}}{B_0} 
&= \frac{a \Omega + b (t+1)}{\Omega} \binom{\Omega}{t+1}  \nonumber \\
&= a \binom{\Omega}{t+1} + b \binom{\Omega-1}{t},  \nonumber \\[0ex]
\frac{\displaystyle\binom{\Omega-1}{t}}{\beta_0} 
&= \frac{a \Omega + b (t+1)}{t+1} \binom{\Omega-1}{t}  \nonumber \\
&= a \binom{\Omega}{t+1} + b \binom{\Omega-1}{t}, 
\end{align}
and the divisibility constraints are met.
\end{proof}


\begin{thm}
\label{Th:Base_scheduling}
For the considered MIMO-CC system, 
one can partition $\SfS^{\CK}$ into $S_0 = \nicefrac{ \textstyle\binom{\displaystyle \Omega}{\displaystyle t+1}}{\displaystyle B_0}$ supersets $\tilde{\SfS}^{\CK}(\tilde{s})$, $\tilde{s} \in [S_0]$\color{black}, 
such that for every $\tilde{s} \in [S_0]$\color{black}, 
\begin{equation}
\begin{aligned}
&\bigcup_{\CT \in \tilde{\SfS}^{\CK}(\tilde{s})} \CT = \CK, \\ 
&\Big|\big\{\CT \in \tilde{\SfS}^{\CK}(\tilde{s}) \,\big|\, k \in \CT \big\}\Big| = \beta_0, 
\quad \forall k \in \CK.
\end{aligned}
\end{equation}\color{black}
In other words\color{black}, user $k$ appears in exactly $\beta_0$ distinct
sets $\CT \in \tilde{\SfS}^{\CK}(\tilde{s})\color{black}$.
\end{thm}

\begin{proof}
The proof is based on two existing theorems on hypergraph factorization  in~\cite{baranyai1974factrization,Bahmanian2014ConnectedTheorem}. By definition, a hypergraph $(\CV,\SfE)$ consists of a finite set of vertices $\CV$ and an edge multi-superset $\SfE$, where every edge $\CE \in \SfE$ is itself a multi-subset of $\CV$. For a positive integer $r$, an $r$-factor in a hypergraph $(\CV,\SfE)$ is a spanning $r$-regular sub-hypergraph of $(\CV,\SfE)$, i.e., a hypergraph with the same vertex set $\CV$ but with an edge superset $ \hat{\SfE} \subseteq \SfE$ such that every vertex in $\CV$ is included in exactly $r$ edges $\CE \in \hat{\SfE}$. The $r$-factorization of $(\CV,\SfE)$ is then defined as the partitioning of $\SfE$ into multiple equal-sized sub-supersets 
\begin{equation*}
 \hat{\SfE}_i, \ i \in \Big\{1,\cdots,\displaystyle \frac{\displaystyle|\SfE|}{\displaystyle|\hat{\SfE}_i|} \Big\},
\end{equation*}
such that every hypergraph $(\CV,\hat{\SfE}_i)$ is an $r$-factor of $(\CV,\SfE)$. 

For a positive integer $h$, the hypergraph $(\CV,\SfE)$ is said to be $h$-uniform if $|\CE| = h$ for each $\CE \in \SfE$. A complete $h$-uniform hypergraph $K_n^h$ is defined as a hypergraph where $|\CV| = n$ and $\SfE$ includes every subset of $\CV$ with size $h$. The well-known Baranyai theorem in~\cite{baranyai1974factrization} states that ``if $\tfrac{\displaystyle n}{\displaystyle h}$ is an integer, $1$-factorization of $K_n^h$ is indeed possible''. The Baranyai theorem was later extended in numerous works~\cite{Bahmanian2014ConnectedTheorem,bahmanian2018extending}, among which, in~\cite{Bahmanian2014ConnectedTheorem} it was shown that ``$K_n^h$ has a connected $\frac{\displaystyle h}{\displaystyle\mathrm{gcd}(n,h)}$-factorization.''

Now, to prove Theorem~\ref{Th:Base_scheduling}, let us first consider the case $\beta_0 = 1$, i.e.,  \(\mathrm{gcd}(\Omega, t+1) = t+1\). Consider the complete $(t+1)$-uniform hypergraph $K_{\Omega}^{t+1}$, where the set of vertices is the same as the target user set $\CK$ and the set of edges includes every selection of users from $\CK$ with size $t+1$. Clearly, for this hypergraph, the superset of edges $\SfE$ is the same as $\SfS^{\CK}$. The original Baranyai theorem~\cite{baranyai1974factrization} states that $K_{\Omega}^{t+1}$ has a 1-factorization, and as each 1-factor should span the vertex set and each edge has a size of $t+1$, the number of edges in each 1-factor is $ \tfrac{\displaystyle\Omega}{\displaystyle t+1} = B_0$. As a result, the total number of 1-factors is 
\begin{equation}
   \frac{\displaystyle |\SfE|}{\displaystyle B_0} = \frac{\displaystyle|\SfS^{\CK}|}{\displaystyle B_0} = \displaystyle\frac{ 
   \displaystyle\binom{\displaystyle\Omega}{\displaystyle t+1}}{\displaystyle B_0} 
\end{equation}

Next, we consider the general case when $\mathrm{gcd}(\Omega, t+1) \neq t+1$. Again, starting from the complete $(t+1)$-uniform hypergraph $K_{\Omega}^{t+1}$, we can use the extension of the Baranyai's theorem in~\cite{Bahmanian2014ConnectedTheorem} to ensure that $K_{\Omega}^{t+1}$ has a $\tfrac{\displaystyle t+1}{\displaystyle \gcd(\Omega,t+1)} = \beta_0$-factorization. Clearly, as each $\beta_0$-factor spans the whole vertex set and each vertex appears exactly $\beta_0$ times, each $\beta_0$-factor provides us with one desired superset $\tilde{\SfS}^{\CK}(\tilde{s})\color{black}$. Moreover, as the number of vertices is $\Omega$, each vertex appears $\beta_0$ times, and the size of each edge is $t+1$, the number of edges in a $\beta_0$-factor is 
\begin{equation}
    \frac{\displaystyle\Omega \beta_0}{\displaystyle(t+1)} = \frac{\displaystyle \Omega}{\displaystyle \gcd(\Omega,t+1)} = B_0
\end{equation}
As a result, the total number of $\beta_0$-factors is 
\begin{equation}
\frac{\displaystyle|\SfE|}{\displaystyle B_0} = \frac{\displaystyle \binom{\displaystyle \Omega}{\displaystyle t+1}}{\displaystyle B_0},
\end{equation}
and the proof is complete.
\end{proof}

\begin{remarknum}
    The proof of Theorem~\ref{Th:Base_scheduling} only shows the existence of the intended partitioning of codeword indices. In order to build such a partitioning, one may use exhaustive search, heuristic solutions, or existing algorithms that are applicable under particular constraints. For example, if $t=1$ and $\Omega$ is even, the partitioning problem reduces to the well-known round robin tournament scheduling which has been thoroughly studied in the literature~\cite{harary1966theory}. Also, for the slightly more general case where $t$ can take any value but $\gcd(\Omega,t+1) = t+1$, the proof of the original Baranyai theorem~\cite{baranyai1974factrization,rasmussen2008round} can be used to design an efficient partitioning algorithm. A detailed description of this solution is provided in~\cite{brouwer1979uniform}. 
\end{remarknum}


\noindent\textbf{Base scheduling.} 
We first find the supersets $\tilde{\SfS}^{\CK}(\tilde{s})\color{black}$, $\tilde{s} \color{black}\in [S_0]$ using Theorem~\ref{Th:Base_scheduling}, and then, we simply design $S_0$ transmission vectors $\Bx_{\CK}(s)\color{black}$ as 
    $\Bx_{\CK}(s) = \sum_{\CT\in \tilde{\SfS}^{\CK}(\tilde{s})}  \Bw_{\CT} X_{\CT}$\color{black}.
Clearly, the base scheduling requires $S_0$ transmit intervals for every subset $\CK$ of users with $|\CK| = \Omega$, but the subpacketization is not affected. 

\begin{thm}\label{Th:Scheduling Mechanism}
For given $\beta_0$ and $S_0$ and for two general integers $\delta$ and $\eta$ satisfying $\frac{\displaystyle \delta S_0}{\displaystyle \eta} \in \mathbb{N}$, a set of feasible $\beta$ values can be built as

\begin{align} 
\CB_{\Omega} = \Bigg\{ \eta \beta_0 \,\Bigg|\, 
&\frac{\delta S_0}{\eta} \in \mathbb{N}, \nonumber \\[-2.5ex]
&\eta \leq \min\left( \frac{L S_0}{1 + (\Omega - t - 1) S_0 \beta_0},\ 
\frac{G}{\beta_0} \right) 
\Bigg\}.
\end{align}\label{eq:lin_fullmc_mac}
In other words, for each $\beta \in \CB_{\Omega}$,
one could build a linear CC scheme comprising only XOR codewords of size $t+1$, such that with each transmission, each user in the target user set $\CK$ with $|\CK| = \Omega$ can decode $\beta$ parallel streams using a linear receiver.
\end{thm}
\begin{figure}[t]
        \hspace{-1.65mm}\centering
        \includegraphics[height = 2.90cm]{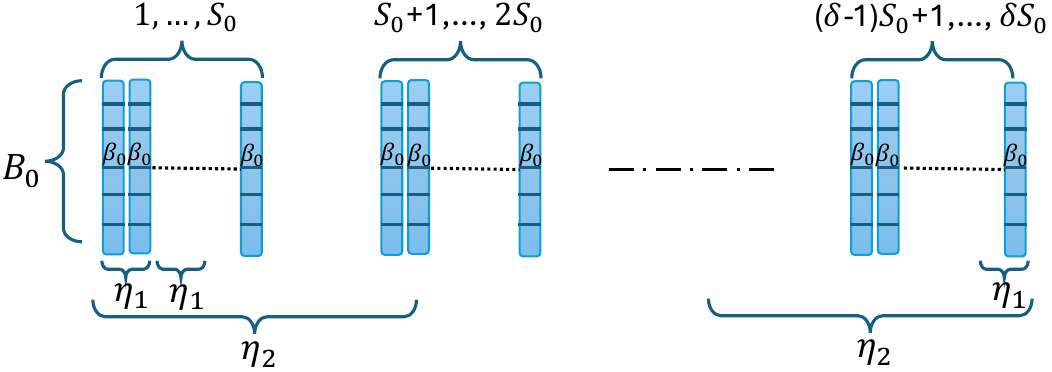} \label{fig:schduling}
    \caption{MIMO-CC multicast scheduling: a base scheduling block of size \(B_0 \times S_0\) is repeated $\delta$ times, and $\eta$ columns are selected from the resulting table for each interval, with two arbitrary options for $\eta$.
    \color{black}}\label{fig:ISIT_sysm}
\end{figure}
\begin{proof}
According to Theorem~\ref{Th:Base_scheduling}, 
one could partition $\SfS^{\CK}$ into $S_0$ supersets $\tilde{\SfS}^{\CK}(\tilde{s})\color{black}$ such that for every $\tilde{s} \in [S_0]$\color{black}, $\bigcup_{\CT \in \tilde{\SfS}^{\CK}(\tilde{s})} \CT = \CK\color{black}$ and each user $k \in \CK$ appears in exactly $\beta_0$ sets $\CT \in \tilde{\SfS}^{\CK}(\tilde{s})\color{black}$. Let us consider one such partitioning and build a $B_0 \times S_0$ table where the column $\tilde{s}\color{black}\in [S_0]$ of the table includes all index sets $\CT \in \tilde{\SfS}^{\CK}(\tilde{s})\color{black}$. By definition, each user $k \in \CK$ appears exactly $\beta_0$ times in each column. Now, assume we first concatenate $\delta$ copies of this table, where $\delta$ can be any integer, to get a larger table of size $B_0 \times \delta S_0$ (in practice, this means increasing the subpacketization by a factor of $\delta$ to avoid retransmission of the same data), and then, we again partition the resulting table into smaller tables of size $B_0 \times \eta$, where the integer parameter $\eta$ is selected such that $\nicefrac{\displaystyle \delta S_0}{\displaystyle \eta}$ is also an integer. This concatenation and partitioning process is shown in Fig.~\ref{fig:ISIT_sysm}.

By definition, each user $k \in \CK$ appears exactly $\eta \beta_0$ times in each resulting small table. Let us use $\hat{\SfS}^{\CK}(\hat{s})$, $\hat{s} \in \big[\nicefrac{\displaystyle \delta S_0}{\displaystyle \eta}\big]$ to denote the multi-superset including all the codeword indices in the $\hat{s}$-th small table (we need a multi-superset as there could be repetition in codeword indices if $\eta > S_0$). Then, one could build $\nicefrac{\displaystyle \delta S_0}{\displaystyle \eta}$ transmission vectors $\Bx_{\CK}(\hat{s})$ as follows

\begin{equation}
\begin{aligned}\label{sig_model_mc}
   \displaystyle \Bx_{\CK}(\hat{s}) = \sum\limits_{\CT\in \hat{\SfS}^{\CK}(\hat{s})}  \Bw_{\CT}^{\hat{q}} X_{\CT}^{\hat{q}}\color{black},  
\end{aligned}
\end{equation}
where the super index ${\hat{q}}$ increases sequentially and is used to distinguish between the codewords (and beamformers) with the same index $\CT$. Each transmission vector $\Bx_{\CK}(\hat{s})\color{black}$ delivers exactly $\eta \beta_0$ subpackets to each user $k \in \CK$ using the codewords of size $t+1$. 

Clearly, a necessary condition for linear decodability of the transmission vectors $\Bx_{\CK}(\hat{s})\color{black}$ is $\eta \beta_0 \le G$.
However, as discussed in the proof of Theorem~\ref{Th:DoF}, one should also ensure that the number of subpackets $W_{\CP,k}^q$ with the same packet index $\CP$ received by each user $k \in \CK$ is constrained by the remaining spatial multiplexing order at that user (i.e., the rank of $\mathrm{Null}\big(\bar{\BH}_{\CP,k}(\hat{s})\big)\color{black}$, where $\bar{\BH}_{\CP,k}(\hat{s})\color{black}$ is defined in~\eqref{eq:equi_intf_chan}). In the proposed scheduling mechanism (see Fig.~\ref{fig:ISIT_sysm}), the number of subpackets with the same packet index delivered to a user is given by $\lceil \nicefrac{\displaystyle \eta}{\displaystyle S_0} \rceil$. 
Following the procedure outlined in the proof of Theorem~\ref{Th:DoF}, the rank of $\mathrm{Null}\big(\bar{\BH}_{\CP,k}(\hat{s})\big)\color{black}$ is 

\begin{equation}
\begin{aligned}
\mathrm{nullity}\big(\bar{\BH}_{\CP,k}(\hat{s})\big)  &= L - \mathrm{rank}\big(\bar{\BH}_{\CP,k}(\hat{s}) \big) \\[1ex]&= L - ({\Omega-t-1})\eta\beta_0.
\end{aligned}
\end{equation}
As a result, for linear decodability, the following condition should hold

\begin{equation}
         \Big\lceil \frac{\eta}{{S_0}} \Big\rceil {\leq} L-(\Omega-t-1)\eta\beta_0 \Leftrightarrow{}
         \eta \leq \!\!\Bigg\lfloor\!\!\frac{LS_0}{1+(\Omega-t-1)\beta_0 S_0}\!\!\Bigg\rfloor, 
\end{equation}
which 
completes the proof.
\end{proof}

\begin{remarknum}\label{rem: full_Mc_mlstr}
    \color{black}
    When the transmitted data is split into multiple parallel sub-streams, two approaches are possible. One option is to strictly subpacketize the file into fixed-size subpackets, enforcing max-min fairness per sub-stream. Alternatively, we can adopt flexible spatial splitting, in which the encoded codeword is divided into sub-streams of arbitrary sizes determined by their allocated rates in the symmetric rate optimization problem~\cite{naseritehrani2023low}. This avoids introducing additional subpacketization in the bit domain and expands the feasible rate region. Nevertheless, in practice, to maintain flexible splitting without introducing significant overhead or impractically small sub-streams, the file size \( F \) should be sufficiently large to accommodate fine-grained splitting.
\end{remarknum}

\begin{exmp}
\color{black} 
\label{Example_subsection}
In this example, we review how the proposed scheduling in Theorem~\ref{Th:Scheduling Mechanism} could be applied to an example MIMO-CC network of $K \ge 10$ users with $(L,G,t) = (10,3,1)$.
Applying the DoF analysis from~\eqref{eq:total_DoF}, it can be verified that $\textrm{DoF}_{} = 15$ is achievable in this network by setting $\Omega = \Omega^* = 5$ and $\beta = \beta^*=3$. However, this DoF could only be achieved by signal-domain processing for this particular scenario. Now, considering the bit domain transmission of XOR's, we investigate the feasible pairs of $\Omega$ and $\beta$ obtained by Theorem~\ref{Th:Scheduling Mechanism} for an example subset of $\Omega \in\{2,4,5,7\}$, chosen to showcase distinct scheduling results (for example, $\Omega\in\{3,6\}$ cases are omitted as they can be shown to result in similar scheduling solutions as $\Omega\in\{4,5\}$, respectively).
For notational simplicity, we remove brackets and commas when explicitly mentioning the codeword index sets $\CT$.



$\bullet$ $\Omega = 2$: In this case, $\gcd(\Omega,t+1) =  2$ and $\beta_0 = B_0 = S_0 = 1$. According to the Theorem~\ref{Th:Scheduling Mechanism}, the feasible set $\CB_{\Omega = 2}$ includes every integer $\eta$ such that $\eta\leq \min( \nicefrac{G}{\beta_0}=3,\lfloor \frac{10}{1+(2-1-1)\times1} \rfloor=10)$, and $\nicefrac{\displaystyle \delta}{\displaystyle \eta}$ is an integer for some integer $\delta$ (naturally, we are interested in the smallest $\delta$ to avoid unnecessary extra processing and subpacketization). As a result, we have $\CB_{\Omega=2} = \{1,2,3\}$, corresponding to total parallel stream counts of $\{2,4,6\}$, respectively, and the subpacketization may also increase proportionally to the selected $\beta$ value (it could be avoided as discussed in Remark~\ref{rem: full_Mc_mlstr}). For example, if we select $\eta=\delta=2$, the transmission vectors for $\CK=\{1,2\}$ will be $(A_2 \oplus B_1)^1\Bw_{12}^1$ and $(A_2 \oplus B_1)^2\Bw_{12}^2$.


$\bullet$ $\Omega = 4$: In this case, $\gcd(\Omega,t+1) = 2$, $\beta_0 = 1$, $B_0 = 2$, and $S_0 = 3$. Let us assume $\CK = [4]$. Then, one could select the index supersets $\tilde{\SfS}^{\CK}(\tilde{s})$ in Theorem~\ref{Th:Base_scheduling} as
    $\tilde{\SfS}^{\CK}(1) = \{ 12,34 \}$, $\tilde{\SfS}^{\CK}(2) = \{ 13,24 \}$, and $\tilde{\SfS}^{\CK}(3) = \{ 14,23 \}$\color{black},
corresponding to transmission vectors $\Bx_{\CK}(\tilde{s})\color{black}$ as

\begin{equation*} \label{eq:sample_trans_omega4}
\begin{aligned}
    \Bx_{\CK}(1) &= (A_2 \oplus B_1) \Bw_{12} + (C_4 \oplus D_3) \Bw_{34}, \\
    \Bx_{\CK}(2) &= (A_3 \oplus C_1) \Bw_{13} + (B_4 \oplus D_2) \Bw_{24}, \\
    \Bx_{\CK}(3) &= (A_4 \oplus D_1) \Bw_{14} + (B_3 \oplus C_2) \Bw_{23},
\end{aligned}
\end{equation*}
\color{black}
respectively. According to Theorem~\ref{Th:Scheduling Mechanism}, the feasible set $\CB_{\Omega=4}$ includes every integer $\eta$ such that $\eta \leq \min( 3 ,\lfloor\frac{10\times 3}{1+(4-1-1)\times 3}\rfloor =  4 )$ and $\nicefrac{\displaystyle 3\delta}{\displaystyle \eta} \in \mathbb{N} $ for some $\delta \in \mathbb{N}$. This results in $\CB_{\Omega=4} = \{1,2,3\}$, corresponding to total parallel stream counts of $\{4,8,12\}$, respectively. For example, if we select $\eta=3$ and $\delta=1$, we can simply transmit a superposition of all the transmission vectors $\Bx_{\CK}(1)$-$\Bx_{\CK}(3)\color{black}$   in~\eqref{eq:sample_trans_omega4} without any need to increase the subpacketization (as $\delta=1$), and the users can employ a linear receiver to extract all the required terms.

$\bullet$ $\Omega=5$: In this case, $\gcd(\Omega,t+1) = 1$ and $\beta_0 = 2$, $B_0 = 5$, and $S_0 = 2$. Let us assume $\CK = [5]$. Then, one could select the index supersets $\tilde{\SfS}^{\CK}(\tilde{s})\color{black}$ in Theorem~\ref{Th:Base_scheduling} as
\begin{align*}
    \tilde{\SfS}^{\CK}(1) = \{12,23,34,45,15\}, \\
    \tilde{\SfS}^{\CK}(2) = \{13,24,35,14,25\},
\end{align*}\color{black}
and, for example, the transmission vector corresponding to $\tilde{\SfS}^{\CK}(1)$ is given as
\begin{equation*}
    \begin{aligned}
        \Bx_{\CK}&(1) = (A_2 \oplus B_1) \Bw_{12} + (B_3 \oplus C_2) \Bw_{23} + \\ & (C_4 \oplus D_3) \Bw_{34} + (D_5 \oplus E_4) \Bw_{45} + (E_1 \oplus A_5) \Bw_{15}.
    \end{aligned}
\end{equation*}\color{black}
According to Theorem~\ref{Th:Scheduling Mechanism}, $\CB_{\Omega=5}$ includes every integer $2\eta$ such that $\eta \leq \min(3, \lfloor\frac{10\times 2}{1+(5-1-1)\times 2\times 2}\rfloor = 1) = 1$ and $\nicefrac{\displaystyle 2\delta}{\displaystyle \eta}$ is an integer for some $\delta \in \mathbb{N}$. As a result, only $\CB_{\Omega=5}\! =\! \{2\}$ is possible given the linear decodability constraint,
corresponding to a total of 10 parallel streams.


$\bullet$ $\Omega = 7$: In this case, $\gcd(\Omega,t+1) = 1$ and $\beta_0 = 2$, $B_0 = 7$, and $S_0 = 3$. While it is possible to write down the base scheduling, from Theorem~\ref{Th:Scheduling Mechanism} we can see that $\CB_{\Omega = 7}$ includes every integer $2\eta$ such that $\eta \leq \min(3 , \lfloor\frac{10\times 3}{1+(7-1-1)\times 3\times 2}\rfloor =  0)$ and $\nicefrac{\displaystyle 3\delta}{\displaystyle \eta}$ is an integer for some $\delta$. Clearly, in this case, $\CB_{\Omega = 7} = \varnothing$, and there exists no scheduling with linear decodability.

\end{exmp}
\color{black}




\section{Simulation Results}
\label{section:Simulations}

Numerical results are generated for various combinations of network parameters $t$, $L$, $G$, and delivery parameters $\Omega$, $\beta$, to compare different transmission strategies. Without loss of generality, the network size is set to $K=20$ users unless specified otherwise. Channel matrices are modeled as i.i.d. complex Gaussian, and the SNR is defined as $\frac{P_T}{N_0}$, where $P_T$ is the power budget at the transmitter and $N_0$ denotes the fixed noise variance. Throughout this section, the keywords UC and MC refer to full unicast scheduling (Theorem~\ref{Th:DoF} and~\ref{Th:linear_achievablity} in Section~\ref{section:DoF}) and full multicast scheduling (Theorem~\ref{Th:Scheduling Mechanism} in Section~\ref{section:DoF_MC}), respectively.
Moreover, Ext-MS denotes the extended multi-server scheme explained in Sec.~\ref{section:DoF_MC}, where design parameters  are \color{black} selected as $t+1 \leq \Omega\leq t+L$ \color{black} and $\beta=\binom{\Omega-1}{t}$ per~\eqref{eq:ldextms}; Ext-Sh refers to the extension of the shared-cache model to the MIMO case in~\cite{salehi2022multi}, achieving a DoF of $G(t+\lfloor \nicefrac{L}{G} \rfloor)$ with optimized Tx-Rx beamforming, where design parameters are selected as $\beta=G$ and $\Omega=t + \lfloor \nicefrac{L}{G} \rfloor $; 
and MU-MIMO denotes the baseline case without any CC technique, but benefiting from local caching gain by serving $L$ users per interval, each receiving one stream.

     \begin{figure}[t]
    \centering
    \resizebox{\columnwidth}{!}{%

   
    \begin{tikzpicture}
    \begin{axis}[
        axis lines = left,
        xlabel = \smaller {Transmitter-side spatial multiplexing gain $L$},
        ylabel = \smaller {Achievable DoF},
        ylabel near ticks,
        xlabel near ticks,
        legend pos = north west,
        legend columns=3,
        legend style = {
            nodes={scale=0.99, transform shape},
            cells={align=center},
            at={(0.01,1)},
            anchor=north west,
            font=\small
        },
        ybar interval=0.8,
        ymin = 0,
        ymax = 100,
        xmin = 8,
        xmax=28,
        ticklabel style={font=\smaller},
        grid=both,
        major grid style={line width=.2pt,draw=gray!30},
        bar width=0.5cm, 
    ]

 \addplot 
    [black,fill=brown!30]
    table
    [x=TxAntenna,y=dofGmimo]
    {Figs/DoF_Scalability1.tex};
    \addlegendentry{\smaller MU-MIMO}

    \addplot 
    [black,fill=black!50]
    table
    [x=TxAntenna,y=dofG8propt1]
    {Figs/DoF_Scalability1.tex};
    \addlegendentry{\smaller $G=8,t=1$}
    
    \addplot 
    [black,fill=red!50]
    table
    [x=TxAntenna,y=dofG16propt1]
    {Figs/DoF_Scalability1.tex};
    \addlegendentry{\smaller $G=16,t=1$}
    
    \addplot
    [black,fill=magenta!40]
    table
    [x=TxAntenna,y=dofG8propt3]
    {Figs/DoF_Scalability1.tex};
    \addlegendentry{\smaller $G=8,t=3$}

  \addplot
    [black,fill=blue!30]
    table
    [x=TxAntenna,y=dofG16propt3]
    {Figs/DoF_Scalability1.tex};
    \addlegendentry{\smaller $G=16,t=3$}

    \end{axis}
    \end{tikzpicture}
    }
    \caption{Achievable DoF of UC and MU-MIMO.
    }
    \label{fig:plot_scal}
      \end{figure}

In Fig.~\ref{fig:plot_scal}, we evaluate the scalability of the achievable DoF (Corollary~\ref{remark1-Dof} in Section~\ref{section:DoF}) in comparison to the baseline MU-MIMO solution. Specifically, we examine how $L$, $G$, and $t$ parameters impact the achievable DoF. 
As can be seen, integrating CC into MIMO communication can significantly boost the achievable DoF, since the CC gain $t$ is scaled by the receiver-side spatial gain $G$ and added to the transmitter-side spatial gain $L$. This is in contrast with baseline MIMO setups, where the DoF is limited by $L$ in the best case.
In addition, this figure confirms that the effect of the channel rank, as the DoF value becomes limited or remains unchanged by $L$ when $L > G$, regardless of any increase in $G$.  
For example, when $L=8$, setting $G=16$ provides no additional benefit compared to $G=8$ for a fixed $t$.

      \begin{figure}[t]
    \centering

    \resizebox{\columnwidth}{!}{%

    \begin{tikzpicture}

    \begin{axis}
    [
    axis lines = center,
    xlabel near ticks,
    xlabel = \smaller {$L$},
    ylabel = \smaller {Achievable DoF},
    ylabel near ticks,
    ymin = 0,
    xmax = 50,
    legend style={
        nodes={scale=0.9, transform shape},
        at={(0,1)}, 
        anchor=north west,
        draw=black, 
        outer sep=2pt, 
        font=\tiny, 
    },
    ticklabel style={font=\smaller},
    grid=both,
    major grid style={line width=.2pt,draw=gray!30},
    ]

   \addplot
    [mark = pentagon, mark size=1.6 , blue!100]
    table[y=dofucTH1t2G8,x=TxAntenna]{Figs/DoF_Prop_SOTA.tex};  
    \addlegendentry{\small UC scheduling-Th~\ref{Th:DoF}};
    \addplot
    [mark = +, mark size=1.6 , red!100]
    table[y=dofmcTH3t2G8,x=TxAntenna]{Figs/DoF_Prop_SOTA.tex};
    \addlegendentry{\small MC scheduling-Th~\ref{Th:Scheduling Mechanism}
    };

    \addplot
    [mark = diamond, mark size=1.6 , black!100]
    table[y=dofextMSt2G8,x=TxAntenna]{Figs/DoF_Prop_SOTA.tex};
    \addlegendentry{\small Ext-MS~\cite{shariatpanahi2016multi}};
    
    \addplot
    [mark = star, mark size=1.6 , green!100]
    table[y=dofwsat2G8,x=TxAntenna]{Figs/DoF_Prop_SOTA.tex};
    \addlegendentry{\small Ext-Sh~\cite{salehi2021MIMO}};

      \addplot
    [mark = star , mark size=1.6 , cyan!100]
    table[y=dofmimoG8,x=TxAntenna]{Figs/DoF_Prop_SOTA.tex};
    \addlegendentry{\small MU-MIMO };

    \end{axis}

    \end{tikzpicture}
    }
    \caption{The achievable DoF of UC, MC schemes, $(G,t)=(8,2)$.}
    \label{fig:Plot13}
      \end{figure}
In 
Fig.~\ref{fig:Plot13}, we evaluate the achievable DoF of the UC
and MC schemes,
highlighting their enhanced flexibility and performance. These approaches are compared against three benchmarks: Ext-MS, Ext-Sh, and the baseline MU-MIMO scheme. 
The results reveal limitations of the Ext-Sh scheme
imposed by the integer constraint $\nicefrac{L}{G}$, restricting its ability to adapt the DoF to the arbitrary system settings properly. 
It can also be seen that the Ext-MS scheme fails to achieve DoF values close to the optimized value in~\eqref{eq:total_DoF}, when \( G < \binom{\Omega-1}{t} \) and linear decodability is imposed. For example, when $(L, G,t) = (30,8,2)$ and under the linear decodability constraint, Ext-MS can only achieve $\textrm{DoF}_{\textrm{MS}} \leq 30$, 
while our proposed UC scheme achieves $\textrm{DoF}^* = 49$.
The figure also illustrates that the achievable DoF of the MC scheme closely tracks that of the UC scheme.
These findings underscore the flexibility of our proposed solutions in accommodating a wide range of system parameters while achieving large DoF gains compared to the state of the art.


      \begin{figure}[t]
    \centering

   \resizebox{\columnwidth}{!}{%

    \begin{tikzpicture}

    \begin{axis}
    [
    axis lines = center,
    xlabel near ticks,
    xlabel = \smaller {SNR [dB]},
    ylabel = \smaller {Symmetric Rate [files/s]},
    ylabel near ticks,
    ymin = 0,
    xmax = 30,
    legend style={
        nodes={scale=0.99, transform shape},
        at={(0,1)}, 
        anchor=north west,
        draw=black, 
        outer sep=2pt, 
        font=\smaller, 
    },
    ticklabel style={font=\smaller},
    grid=both,
    major grid style={line width=.2pt,draw=gray!30},
    ]


    

    
   \addplot
    [dashed,mark = square,  mark options={solid, fill=blue!100}, mark size=2.4, blue!100]
    table[y=L11G8t2MCOmega4beta6,x=SNR]{Figs/Symmetric_rate_Prop_SOTA_Comparison.tex};
    \addlegendentry{\small MC, $(11,8,2,4,6)$
    } 

    \addplot
    [,mark = square, mark size=2.4,mark options={solid, fill=red}, red!100]
    table[y=L11G8t2MIMOOmega11beta1LC,x=SNR]{Figs/Symmetric_rate_Prop_SOTA_Comparison.tex};
    \addlegendentry{\small MU-MIMO, $(11,8,2,11,1)$
    }
    

    \addplot
    [dashed,mark = o,  mark options={solid, fill=blue!100}, mark size=2.4, blue!100]
    table[y=L7G5t2MCOmega5beta3,x=SNR]{Figs/Symmetric_rate_Prop_SOTA_Comparison.tex};
    \addlegendentry{\small MC, $(7,5,2,5,3)$ 
    }
    
 \addplot
    [,mark = o, mark size=2.4,mark options={solid, fill=red}, red!100]
    table[y=L7G5t0MIMOOmega7beta1LC,x=SNR]{Figs/Symmetric_rate_Prop_SOTA_Comparison.tex};
    \addlegendentry{\small MU-MIMO, $(7,5,2,7,1)$
    }


    \addplot
    [dashed,mark = star, mark size=2.4, blue!100]
    table[y=L11G31MCOmega7beta2,x=SNR]{Figs/Symmetric_rate_Prop_SOTA_Comparison.tex};
    \addlegendentry{\small MC, $(11,3,1,7,2)$
    } 


    \addplot
    [,mark = star, mark size=2.4, red!100]
    table[y=L11G3t0MIMOOmega11beta1LC,x=SNR]{Figs/Symmetric_rate_Prop_SOTA_Comparison.tex};
    \addlegendentry{\small MU-MIMO, $(11,3,1,11,1)$
    }

    \end{axis}

    \end{tikzpicture}
    }
    \caption{Symmetric rate of MU-MIMO vs MC for different $(L,G,t,\Omega,\beta)$. 
    }
    \label{fig:COV_ZF_MC_Scalability}
      \end{figure}
      
In Fig.~\ref{fig:COV_ZF_MC_Scalability}, we illustrate the impact of the CC gain $t$ on the symmetric rate in MIMO systems for the MC and baseline MU-MIMO schemes for the following system setups: $(L,G,t) \in \{(11,3,1), (7,5,2), (11,8,2)\}$. For a fair comparison, the scheduling parameters ($\Omega$ and $\beta$) are selected to achieve the best performance within the given SNR range, even if this does not necessarily correspond to utilizing the full number of streams implied by the degrees of freedom (DoF). 
As can be seen, even with a small CC gain of $t\in\{1,2\}$, the MC scheme can
significantly enhance the symmetric rate compared to the baseline MIMO solution throughout the entire SNR range by flexibly choosing the best scheduling option while benefiting from the CC gain.

     \begin{figure}[t]
    \centering

   \resizebox{\columnwidth}{!}{%

    \begin{tikzpicture}

    \begin{axis}
    [
    axis lines = center,
    xlabel near ticks,
    xlabel = \smaller {SNR [dB]},
    ylabel = \smaller {Symmetric Rate [files/s]},
    ylabel near ticks,
    ymin = 0,
    xmax = 30,
    legend style={
        nodes={scale= 0.99, transform shape},
        at={(0,1)}, 
        anchor=north west,
        draw=black, 
        outer sep=2pt, 
        font=\smaller, 
    },
    ticklabel style={font=\smaller},
    grid=both,
    major grid style={line width=.2pt,draw=gray!30},
    ]

    \addplot
    [dashed,mark = square, mark size=2.4 ,  mark options={solid, fill=blue!100}, blue!100]
    table[y=L11G8t2MCOmega4beta6,x=SNR]{Figs/Symmetric_rate_Prop_SOTA_Comparison.tex};
    \addlegendentry{\footnotesize MC-$\Omega=4$, $\beta=6$} 

     \addplot
    [dashed,mark = asterisk, mark size=2.4,  mark options={solid, fill=blue!100}, blue!100]
    table[y=L11G8t2MCOmega6beta3,x=SNR]{Figs/Symmetric_rate_Prop_SOTA_Comparison.tex};
    \addlegendentry{\footnotesize MC-$\Omega=6$, $\beta=3$} 
    
    \addplot
    [dash dot, mark = triangle, mark size=2.4,  mark options={solid, fill=black!90} , black!90]
    table[y=L11G8t2MSOmega4beta3,x=SNR]{Figs/Symmetric_rate_Prop_SOTA_Comparison.tex};
    \addlegendentry{\footnotesize Ext-MS-$\Omega=4$, $\beta=3$}

    \addplot
    [dash dot, mark = x, mark size=2.4 , mark options={solid, fill=red!100}, red!100]
    table[y=L11G8t2WSAUCOmega3beta8,x=SNR]{Figs/Symmetric_rate_Prop_SOTA_Comparison.tex};
    \addlegendentry{\footnotesize 
    Ext-Sh-$\Omega=3$, $\beta=8$}
    

    \addplot
    [ dotted 
    ,mark = +, mark size=2.4, mark options={solid, fill=cyan}, line width = 0.6pt, cyan!100]
    table[y=L11G8t2UCOmega4beta6,x=SNR]{Figs/Symmetric_rate_Prop_SOTA_Comparison.tex};
    \addlegendentry{\footnotesize UC-$\Omega=4$, $\beta=6$}

    \addplot
    [,mark = o, mark size=2.4, mark options={solid, fill=cyan}, line width = 0.5pt, cyan!100]
    table[y=L11G8t2UCOmega6beta3,x=SNR]{Figs/Symmetric_rate_Prop_SOTA_Comparison.tex};
    \addlegendentry{\footnotesize UC-$\Omega=6$, $\beta=3$}

    \end{axis}

    \end{tikzpicture}
    }
    \caption{The symmetric rates of MC, UC, Ext-MS and Ext-Sh with $(L,G,t)=(11,8,2)$.
    }
    \label{fig:Plot_SOTA_MC_UC_MC_MS_WSA_t2}
      \end{figure}
      
In Fig.~\ref{fig:Plot_SOTA_MC_UC_MC_MS_WSA_t2}, we extend the symmetric rate evaluation in Fig.~\ref{fig:COV_ZF_MC_Scalability} by comparing MC and UC schemes with Ext-MS and Ext-Sh for a setup with $(L,G,t)=(11,8,2)$. The figure reveals several key observations: 1) The multicasting gain stemming from the bit-level design of MC and Ext-MS schemes has a significant effect on the symmetric rate at finite SNR. In fact, despite delivering a significantly smaller number of parallel streams, the Ext-MS scheme outperforms Ext-Sh and UC for SNR values smaller than 10$\mathrm{dB}$. This observation aligns with previous results for MISO-CC schemes~\cite{tolli2017multi,salehi2022multi}. 2) For a given scheme, if the SNR is small, it is not desirable to increase the number of parallel streams as much as possible. For example, for the SNR value of 5$\mathrm{dB}$, the MC scheme with $(\Omega,\beta) = (6,3)$ outperforms the $(\Omega,\beta) = (4,6)$ scheduling. 3) Comparing MC with Ext-MS and Ext-Sh, we observe that the larger scheduling space provided by our proposed algorithms significantly enhances performance across the entire SNR range, outperforming the state-of-the-art. 

      \begin{figure}[t]
    \centering

   \resizebox{\columnwidth}{!}{%

    \begin{tikzpicture}

    \begin{axis}
    [
    axis lines = center,
    xlabel near ticks,
    xlabel = \smaller {SNR [dB]},
    ylabel = \smaller {Symmetric Rate [files/s]},
    ylabel near ticks,
    ymin = 0,
    xmax = 30,
    legend style={
        nodes={scale=0.89, transform shape},
        at={(0,1)}, 
        anchor=north west,
        draw=black, 
        outer sep=2pt, 
        font=\tiny, 
    },
    ticklabel style={font=\smaller},
    grid=both,
    major grid style={line width=.2pt,draw=gray!30},
    ]

    \addplot
    [mark = , mark size=2.4, red!100]
    table[y=L10G5t1COVOmega4beta4,x=SNR]{Figs/Symmetric_rate_Prop_SOTA_Comparison.tex};
    \addlegendentry{\small COV-$\Omega=4$, $\beta=4$}

    \addplot
    [dash dot,mark = star, mark size=2.4, red!100]
    table[y=L10G5t1MCOmega4beta4,x=SNR]{Figs/Symmetric_rate_Prop_SOTA_Comparison.tex};
    \addlegendentry{\small MC-$\Omega=4$, $\beta=4$}

     \addplot
    [dash dot,mark = triangle,mark options={solid, fill=red!100}, mark size=2.4, red!100]
    table[y=L10G5t1UCmega4beta4,x=SNR]{Figs/Symmetric_rate_Prop_SOTA_Comparison.tex};
    \addlegendentry{\small UC-$\Omega=4$, $\beta=4$}

    \addplot
    [mark = star, mark size=2.4, black!80]
    table[y=L10G5t1MCOmega6beta2,x=SNR]{Figs/Symmetric_rate_Prop_SOTA_Comparison.tex};
    \addlegendentry{\small MC-$\Omega=6$, $\beta=2$}

      \addplot
    [dash dot, mark = o,mark options={solid, fill=black!80}, mark size=2.4, black!80]
    table[y=L10G5t1MSOmega4beta3,x=SNR]{Figs/Symmetric_rate_Prop_SOTA_Comparison.tex};
    \addlegendentry{\small MC-$\Omega=4$, $\beta=3$}

      \addplot
    [dashed, mark = +,mark options={solid, fill=black!80}, mark size=2.4, black!80]
    table[y=L10G5t1MCOmega3beta4,x=SNR]{Figs/Symmetric_rate_Prop_SOTA_Comparison.tex};
    \addlegendentry{\small MC-$\Omega=3$, $\beta=4$}

    \addplot
      [mark = star, mark size=2.4,mark options={solid, fill=blue!100}, blue!100]
    table[y=L10G5t1MCOmega10beta1,x=SNR]{Figs/Symmetric_rate_Prop_SOTA_Comparison.tex};
    \addlegendentry{\small MC-$\Omega=10$, $\beta=1$} 

       \addplot
    [dashed,mark = o, mark size=2.4,mark options={solid, fill=blue!100},, blue!100]
    table[y=L10G5t1MCOmega5beta2,x=SNR]{Figs/Symmetric_rate_Prop_SOTA_Comparison.tex};
    \addlegendentry{\small MC-$\Omega=5$, $\beta=2$} 

       \addplot
    [dash dot,mark = +,  mark options={solid, fill=blue!100}, mark size=2.4, blue!100]
    table[y=L10G5t1MCOmega2beta5,x=SNR]{Figs/Symmetric_rate_Prop_SOTA_Comparison.tex};
    \addlegendentry{\small MC-$\Omega=2$, $\beta=5$}

    \end{axis}

    \end{tikzpicture}
    }
    \caption{
    The effect of the scheduling decision on the symmetric rate.
    }
    \label{fig:Plot_SOTA_MC_UC_COV_MC_same_DOF}
      \end{figure}
       Fig.~\ref{fig:Plot_SOTA_MC_UC_COV_MC_same_DOF} compares the  symmetric rate performance of different scheduling schemes under the same DoF,
for a setup with $(L,G,t)=(10,5,1)$, which can support DoF of $16$ with both UC and MC approaches. All curves with \(\Omega \times \beta = 16\) exhibit the same slope at high SNR. However, in the UC case, this slope becomes evident only at very high SNR, not yet observable even at 30 dB. While the COV design~\cite{naseritehrani2024multicast} serves as a bound for symmetric rate under a given scheduling, it involves significant complexity. Remarkably, the proposed linear MC beamforming solution closely follows the performance of the COV design, demonstrating its effectiveness with much lower complexity.
In addition, we have considered two pairs of scheduling alternatives, where each pair delivers the same total number of parallel streams ($12$ and $10$) but with different $(\Omega,\beta)$ values. From the figure, it can be observed that the scheduling decision with a smaller $\beta$ outperforms the other. This observation aligns with the discussion in Remark~\ref{underloading}. 
For a given number of transmitted streams, distributing the streams across a larger number of users greatly enhances the symmetric rate performance at finite SNR.
Similarly, adopting the setup ($\Omega$=6, $\beta$=2) instead of ($\Omega$=4, $\beta$=4) below the SNR values of 16dB, benefits from the spatial underloading condition at both the transmitter and receiver sides.

\section{Conclusion}
\label{sec:conclusions}
This paper investigated the application of coded caching (CC) to enhance communication efficiency and performance in MIMO systems. First, the number of users per transmission and the spatial multiplexing order per user were optimized to improve the achievable single-shot degrees of freedom (DoF). Then, a new class of MIMO-CC schemes with maximal multicasting gain for enhanced finite-SNR performance and adhering to linear decodability was introduced. Numerical simulations confirmed the enhanced DoF and improved finite-SNR performance of the new schemes.

\appendices
\section{Linear Beamforming for MIMO-CC}
\label{section:Linear_BF}

Here, we discuss how linear transmit and receive beamformers can be designed for the proposed class of MIMO-CC multicast transmission schemes. The objective is to support multicast transmission to multiple user groups with partially overlapping user sets, using receiver-side processing to separate group-specific streams. The solution builds on the approach in~\cite{mahmoodi2021low}, extending it to accommodate partially overlapping multicast groups.
%
%
%
%
Let us start with the general transmission vector design in~\eqref{sig_model_mc} and ignore the $\hat{s}$ index (the same process is repeated for each transmission). 
Let us define $\CD^{\CK}$ to include all the codewords (i.e., every $X_{\CT}^{\hat{q}}$) in transmit signal $\Bx$. We also define $\hat{\SfS}^{\CK}_k = \{\CT \in \hat{\SfS}^{\CK} \mid k \in \CT \}$, $\CD^{\CK}_k = \{X_{\CT}^{\hat{q}} \in \CD^{\CK} \mid \CT \in \hat{\SfS}^{\CK}_k  \}$ and $\Bar{\CD}^{\CK}_k = \CD^{\CK} \backslash \CD^{\CK}_k$.
Then, the signal received by user~$k \in \CK$ in~\eqref{eq:RX_signal} is
\begin{equation}\label{rx_sig_model_Mcuc_basic}
\begin{array}{l}
\By_k = \BH_k \Big( {\sum_{\CT, \hat{q}: X_\CT^{\hat{q}} \in \CD^{\CK}_k} \Bw_{\CT}^{\hat{q}} X_{\CT}^{\hat{q}} }  \\
\qquad \qquad + {\sum_{\CT, \hat{q}: X_\CT^{\hat{q}} \in \bar{\CD}^{\CK}_k }   \Bw_{\CT}^{\hat{q}} X_{\CT}^{\hat{q}}} \Big) + \Bz_k, %
\end{array}
\end{equation}
where the first and second summations represent the intended and interference signals, respectively. 
Denoting $\Bu_{k,\CT}^{\hat{q}}$ as the receiver beamforming vector for decoding the intended stream $X_{\CT}^{\hat{q}} \in \CD_k^{\CK}$ at user~$k \in \CK$, the corresponding SINR term
$\gamma_{k,\CT}^{\hat{q}}$ is given as:
\begin{equation}
\begin{array}{l}
\label{eq:snrmc}
\gamma_{k,\CT}^{\hat{q}}   = \frac{\displaystyle | {\Bu_{k,\CT}^{\hat{q}}}\herm \BH_k \Bw_{\CT}^{\hat{q}}|^2}{ \hspace{-10mm}{\displaystyle \sum\limits_{{\bar{\CT},  \bar{q}: X_{\bar{\CT}}^{\bar{q}} \in \CD^{\CK} \backslash \{X_{\CT}^{\hat{q}}\}}}} 
\displaystyle|  {\Bu_{k,\CT}^{\hat{q}}}\herm \BH_k \Bw_{\bar{\CT}}^{\bar{{{q}}}} |^2  + N_0 \|\Bu_{k,{\CT}}^{\hat{q}}\|^2}.
\end{array}
\end{equation}
Similarly to~\cite{naseritehrani2024multicast}, we aim to minimize the worst-case delivery time among all users in $\CK$. This is realized by maximizing the minimum achievable rate 
across all partially overlapping groups of $\CT \in \hat{\SfS}_k^{\CK}$, formulated as 
\begin{equation}
    \begin{array}{l}\label{eq:optmization_problem}
     \hspace{-3mm}\max_{\Bw_{\CT}^{\hat{q}}, \Bu_{k,\CT}^{\hat{q}}} \!\!
 \min_{ \CT \in  \hat{\SfS}^{\CK}}  \sum_{{\hat{q}} \in \CQ_{\CT}} \min_{k \in \CT} \log(1 +  \gamma_{k,\CT}^{\hat{q}}), \\ [2ex] 
  \qquad s.t.~~\sum\nolimits_{\CT \in \hat{\SfS}^{\CK}, \hat{q} \in \CQ_\CT } \|\Bw_{\CT}^{\hat{q}}\|^2  
	  \leq P_T ,
    \end{array}
\end{equation}
where $\CQ_\CT = \{ {\hat{q}} \mid X_{\CT}^{\hat{q}} \in {\CD}^{\CK} \}$ and
$P_T$ is the transmission power. 
\color{black}
The optimization problem in ~\eqref{eq:optmization_problem} can be solved by alternate optimization of $\{\Bu_{k,\CT}^{\hat{q}}\}$ and $\{\Bw_{k}^{\hat{q}}\}$. For given $\{\Bw_{k}^{\hat{q}}\}$, the 
rate-optimal $\{\Bu_{k,\CT}^{\hat{q}}\}$, maximizing the objective of~\eqref{eq:optmization_problem} and employed to separate the $\beta$ data terms intended for user~$k$, correspond to (scaled) MMSE receivers~\cite{kaleva2016decentralized}:
\begin{equation}
\begin{array}{l}
\label{eq:receive_beamformer}
\Bu_{k,\CT}^{\hat{q}} = \Big(\BH_k \BW\BW\herm \BH_k\herm + N_0 \BI\Big)^{-1} \BH_k {\Bw}_{\CT}^{\hat{q}} , \\  \qquad \qquad\qquad  \forall \CT \in \hat{\SfS}^{\CK}, k \in \CT, {\hat{q}}\in\CQ_\CT, 
\end{array}
\end{equation}
%
where $\BW = \big[{\Bw}_{\CT}^{\hat{q}}\big]$ is formed by concatenation of all transmit beamforming vectors ${\Bw}_{\CT}^{\hat{q}}$ (for every transmitted stream $X_{\CT}^{\hat{q}} \in \CD^{\CK}$). 
However, for given $\{\Bu_{k,\CT}^{\hat{q}}\}$, the (sub-)optimal solution to $\{\Bw_{k}^{\hat{q}}\}$ can be found through a tailored version of the solution in~\cite{kaleva2016decentralized}, coupled with the iterative KKT-based method in~\cite{mahmoodi2021low}. The details are relegated to~\cite{kaleva2016decentralized, mahmoodi2021low}.

\color{black}

\bibliographystyle{IEEEtran}
\bibliography{conf_short,IEEEabrv,references,whitepaper}
\end{document}